\documentclass[journal]{IEEEtran}
\IEEEoverridecommandlockouts

\makeatletter
\def\markboth#1#2{\def\leftmark{\@IEEEcompsoconly{\sffamily}\MakeUppercase{\protect#1}}%
\def\rightmark{\@IEEEcompsoconly{\sffamily}\MakeUppercase{\protect#2}}}
\makeatother

\usepackage[utf8x]{inputenc}
\usepackage[english]{babel}
\usepackage{ucs}
\usepackage{amsmath}
\usepackage{amsfonts}
\usepackage{amssymb}
\usepackage{amsthm}
\usepackage{array}
\usepackage{mathrsfs}
\usepackage{verbatim}
\usepackage{multirow}
\usepackage{listings}
\usepackage{stfloats}

\usepackage{algorithm}
\usepackage{algorithmic}
\usepackage{url}
  \usepackage{enumerate}
  \usepackage{cite}
\usepackage[usenames,dvipsnames,svgnames,table]{xcolor}
\usepackage{setspace}
\ifCLASSINFOpdf
  \usepackage[pdftex]{graphicx}
  \usepackage{subfigure}
  \usepackage{epstopdf}
   \graphicspath{{./}}
   \DeclareGraphicsExtensions{.pdf,.eps,.ps,.png}
\else
  \usepackage[dvipdf]{graphicx}
  \usepackage{subfigure}
  \graphicspath{{./img/}}
   \DeclareGraphicsExtensions{.eps,.ps}
\fi

\newcommand{\Hb}{\mathbf{H}}

\newcommand{\A}{\mathbf{A}}

\newcommand{\I}{\mathbf{I}}

\newcommand{\R}{\mathbf{R}}

\newcommand{\D}{\mathbf{D}}
\newcommand{\one}{\mathbf{1}}

\newcommand{\x}{\mathbf{x}}

\newcommand{\s}{\mathbf{s}}
\newcommand{\q}{\mathbf{q}}

\newcommand{\uu}{\mathbf{u}}

\newcommand{\bb}{\mathbf{b}}

\newcommand{\w}{\mathbf{w}}

\newcommand{\lambdab}{\mathbf{\boldsymbol{\lambda}}}

\newcommand{\ab}{\mathbf{a}}

\newcommand{\pp}{\mathbf{p}}

\newcommand{\Ex}[2]{{\textnormal{E}_{#1}\left[#2\right]}}

\newtheorem{definition}{Definition}
\newtheorem{lemma}{Lemma}
\newtheorem{theorem}{Theorem}
\newtheorem{remark}{Remark}

  \usepackage{xcolor}
   \definecolor{blueH3}{rgb}{0,.5,1}
   \definecolor{blueH2}{rgb}{0,0.25,0.75}
   \definecolor{blueH1}{rgb}{0,0,0.5}
   \definecolor{grayOldText}{rgb}{.5,.5,.5}
   \definecolor{VCobalt}{HTML}{005682}
   \definecolor{TZTeal}{HTML}{008080}
   \definecolor{KYJade}{HTML}{008151}
   \definecolor{ARust}{HTML}{a10000}

\newcounter{MYtempeqncnt}
\newcommand{\inlineeqnum}{\refstepcounter{equation}~~\mbox{(\theequation)}}

\title{Optimal Link Scheduling in Millimeter Wave Multi-hop Networks with MU-MIMO radios.}

\author{
Felipe G\'omez-Cuba$^1$, \emph{Member, IEEE}, Michele Zorzi$^1$, \emph{Fellow, IEEE}, 
\thanks{This work extends results presented at the IEEE Information Theory and Applications Workshop (ITA), La Jolla, February 2016~\cite{gomezITAoptimal}. $^1$: Felipe G\'omez-Cuba and Michele Zorzi are with Dipartimento di Ingegneria dell'Informazione, University of Padova, Via Gradenigo 6/b, 35131 - Padova Italy, (e-mail: {\tt \{gomezcuba,zorzi\}@dei.unipd.it}). This project has received funding from the European Union's Horizon 2020 research and innovation programme under the Marie Sk\l{}odowska-Curie grant agreement No 704837.}
}
\begin{document} 
\maketitle
\markboth{DRAFT}{DRAFT}
\begin{abstract}
This paper studies the maximum throughput achievable with optimal scheduling in multi-hop mmWave picocellular networks with Multi-user Multiple-Input Multiple-Output (MU-MIMO) radios. MU-MIMO enables simultaneous transmission to multiple receivers (Space Division Multiplexing) and simultaneous reception from multiple transmitters (Space Division Multiple Access). The main contribution is the extension to MU-MIMO of the Network Utility Maximization (NUM) scheduling framework for multi-hop networks. We generalize to MU-MIMO the classic proof that Maximum Back Pressure (MBP) scheduling is NUM optimal. MBP requires the solution of an optimization that becomes harder with MU-MIMO radios. In prior models with one-to-one transmission and reception, each valid schedule was a \textit{matching} over a graph. However, with MU-MIMO each valid schedule is, instead, a \textit{Directed Bipartite SubGraph} (DBSG). In the general case this prevents finding efficient algorithms to solve the scheduler. We make MU-MIMO MBP scheduling tractable by assuming fixed power allocation, so the optimal scheduler is the Maximum Weighted DBSG. The MWDBSG problem can be solved using standard Mixed Integer Linear Programing. We simulate multi-hop mmWave picocellular networks and show that a MU-MIMO MBP scheduler enables a 160\% increase in network throughput versus the classic one-to-one MBP scheduler, while fair rate allocation mechanisms are used in both cases.
\end{abstract}

\begin{IEEEkeywords}
5G, Millimeter Wave, Beamforming, Space Division Multiplexing, Dynamic Duplexing, Scheduling, Network Utility Maximization
\end{IEEEkeywords}

\section{Introduction}
\label{sec:introduction}

%\begin{spacing}{1.56}
Millimeter wave (mmWave) frequency bands ($30$-$300$ GHz) have been proposed as a candidate to satisfy the spectrum demands of fifth generation (5G) cellular wireless networks \cite{rappaport2013,RanRapEr:14}. mmWave frequencies offer a $200\times$ increase in available spectrum, allowing channel bandwidths in the GHz range. In addition, thanks to the short wavelength, mmWave radios can pack dozens of antennas even in compact mobile devices. The main drawbacks of mmWave are that pathloss, absorption, and blockages are more severe than in traditional cellular bands. Although the large antenna arrays compensate some of the harsh propagation conditions, the typical range of mmWave devices is about $100$ m. This means that mmWave cellular systems require ultra-dense picocellular deployments with numerous Access Points (AP) per unit area. Since fiber backhauling is costly at such densities, many of those APs will have wireless Integrated Access and Backhaul (IAB) in the same mmWave band \cite{8647977,polese2019integrated}, acting as Relay Nodes (RN) rather than Base Stations (BS)  as shown in Fig. \ref{fig:picocell}.

\begin{figure}
 \centering
 \includegraphics[width=.7\columnwidth]{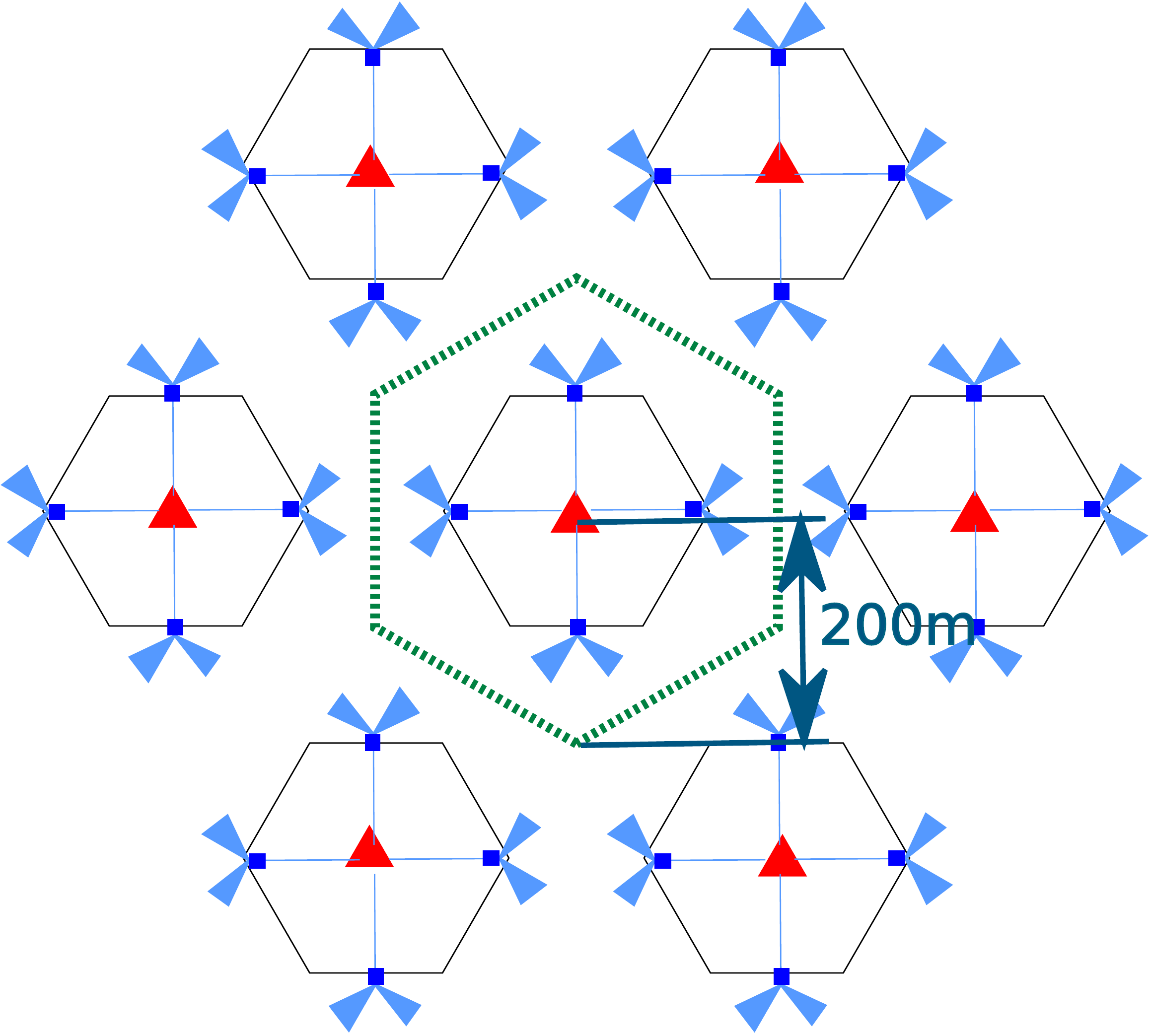}
 \caption{Multi-hop mmWave picocellular network with BSs (red triangles), IAB RNs (blue squares), and UEs (not depicted) uniformly distributed in each BS's ``cell'' with radius $200$ m, making inter-BS distance $346$ m.}
 \label{fig:picocell}
\end{figure}

Current cellular systems achieve spatial multiplexing rate gains at the physical layer through Multi-User MIMO (MU-MIMO) techniques \cite{Gesbert2007a,hoydis2011massive,Bjornson2016} that use multiple antennas to simultaneously transmit towards multiple receivers using Space Division Multiplexing (SDM) or simultaneously receive from multiple transmitters using Space Division Multiple Access (SDMA) \cite{Gesbert2007a,hoydis2011massive,Bjornson2016}.

As an evolution of current cellular systems, mmWave picocellular systems with large antenna arrays are a prime candidate to exploit MU-MIMO at the physical layer, while IAB RNs are required due to the limited ranges. Therefore, a complete scheduling model for future mmWave picocellular systems with IAB RNs should feature both multi-hop and MU-MIMO, and optimize the entire system treated as a mesh multi-hop network comprising wired-backhaul BSs, IAB RNs, and User Equipment (UE) \cite{russellDynamic,juanScheduling}, and where simultaneous transmissions and receptions are possible. Unfortunately, most models in the literature considered either single-hop MU-MIMO \cite{Gesbert2007a,hoydis2011massive,Bjornson2016,Kwon2016b} or multi-hop omnidirectional single-antenna systems \cite{Tass,Eryilmaz2007,Kelly1997,Kelly1998,ModianoPower}. The former model can represent conventional cellular networks with single-hop topologies. On the other hand, the latter model represents the MAC scheme in multi-hop sensor and ad hoc WiFi networks, which is based on a single-antenna omnidirectional collision model. Thus, the main goal of our paper is to study the novel scheduling framework with simultaneous support for multi-hop and MU-MIMO as is necessary in future mmWave IAB picocellular systems.

In this paper, we propose and analyze a complete model for scheduling in a mesh-topology multi-hop mmWave picocellular network with MU-MIMO at both transmitters and receivers. Our model characterizes the long term throughput-optimal scheduling and rate control for arbitrary MU-MIMO mesh networks using the Network Utility Maximization (NUM) framework \cite{Tuto,juanScheduling}. We use this model to evaluate the capacity of multi-hop mmWave picocellular mesh networks with IAB RNs under optimal scheduling. We extend to MU-MIMO a classic analysis that studied a Markov process representing the state of the network, proving that when the scheduler adopts a Maximum Back Pressure (MBP) policy and the source rates are controlled by an Adaptive NUM Congestion Control (ANCC) scheme, the network is stable  (all network states with non-zero probability correspond to finite queues) and the long-term average user rates converge to the maximum of a certain utility metric, which we can choose following fairness or service differentiation criteria. The MBP policy selects the links that have the highest ``queue pressure,'' i.e., difference in the queue length at the transmitter and the receiver.

Even though our theoretical result proves that MU-MIMO MBP scheduling is throughput and NUM optimal for any arbitrary MU-MIMO multi-hop network, this still leaves as a challenge the implementation of the scheduler as a non-convex optimization problem. In classic ad hoc network NUM literature \cite{Tass,Eryilmaz2007,Kelly1997,Kelly1998,ModianoPower} and in the one-to-one-constrained mmWave network model \cite{juanScheduling}, each possible schedule was a matching\footnote{A matching in a graph is defined as a set of edges that have no vertices in common.} over the graph that represents the network. Thus, in prior work, the Maximum Weighted Matching (MWM) algorithm from classical graph theory could be exploited to implement the MBP scheduler with polynomial complexity. Nevertheless, in the MU-MIMO case, due to the fact that transmissions/receptions to/from multiple receivers/transmitters at once are enabled, all valid schedules are not ``matchings'' but rather ``Directed Bipartite SubGraphs'' (DBSG) of the network graph. In the most general case, the link rates vary when different DBSGs of the network are selected due to power allocation at the transmitters, and the implementation of the optimum MBP scheduler is challenging.

Besides the theoretical contribution of generalizing the proof to MU-MIMO, the second main contribution in this paper is the evaluation of the capacity of MU-MIMO mmWave multi-hop picocellular networks with IAB RNs. For this we propose an assumption that makes the scheduler a tractable optimization problem, while being consistent with existing mmWave transceiver circuit designs. Specifically, we assume that power allocation in all transmissions is pre-selected to a fixed value, which is realistic when hybrid analog-digital beamforming ports with independent power amplifiers are used in the mmWave transceiver. With this assumption, link rates become fixed weights, and the MU-MIMO MBP scheduler is reduced to a \textbf{Maximum Weighted DBSG} problem which is tractable using Mixed Integer Linear Programming (MILP).

In our capacity evaluation we compared our MU-MIMO - MILP mmWave multi-hop picocellular scheduling model versus the prior one-to-one - MWM mmWave multi-hop picocellular scheduling model in \cite{juanScheduling}. We found that over a simulation campaign of 50 randomly generated networks with 1 BS, 4 RNs and 10 UEs each, on average the introduction of MU-MIMO physical layer techniques enabled a 160\% increase in sum rate in the cell, while fairness among the users remained similar. Thus, support of MU-MIMO multi-hop scheduling can dramatically increase the performance of future mmWave picocellular networks with IAB RNs. The main engineering lesson of our analysis is that the capacity of future IAB mmWave cellular systems can be dramatically increased by combining MU-MIMO and optimal multi-hop scheduling.

The rest of this paper is organized as follows. Section \ref{sec:model} describes the system model. Section \ref{sec:optimization} describes the general NUM problem statement and the general theoretical result. Section \ref{sec:mmWaveImplementation} describes the simplified implementation of the MU-MIMO MBP optimal scheduler for mmWave systems. Section \ref{sec:numeric} provides simulation results and compares the performance of our MU-MIMO scheduling model versus the prior one-to-one scheduling models, showing that exploiting MU-MIMO can dramatically increase network throughput. Finally, Section \ref{sec:conclusion} concludes the paper.

\subsection{Related Work}

MmWave propagation is studied in \cite{rappaport2013,RanRapEr:14,Rappaport2015}. Hybrid analog-digital antenna array transceiver architectures and MU-MIMO schemes are discussed in\cite{Hur2013,Rappaport2014mimo,Samsung2014,Kutty2015}. Neighbor discovery and channel estimation for mmWave have been extensively studied, see for example \cite{Barati2015,Xiao2017a,gomezcuba2019CSicc}. Single-hop Urban Micro-Cell capacity evaluations are implemented in \cite{Akdeniz2013,Akdeniz2014}.

Cellular networks first introduced multi-hop RNs in the Third Generation Partnership Project (3GPP) Long Term Evolution - Advanced (LTE-A) Release 10 \cite{LTE36216PHYrelay}. In LTE-A, RN scheduling is quite restricted, as RNs can only be connected to one BS and they need to maintain backwards-compatibility with the fixed uplink (UL) and downlink (DL) parts of the LTE frame, which was designed for single-hop systems \cite{fgomez2014improvedrelaying}. Significant work has been done in tractable cellular capacity evaluations with relaying \cite{andrews2014selfbackhaul} using stochastic geometry, however these models have assumed a simplified two-hop tree topology as in LTE-A relaying, where scheduling is drastically simplified to the mere adjustment of the static time-sharing factor between the two hops.

Unlike in the sub-$6$ GHz bands, in mmWave the severe pathloss and directivity typically result in reduced interference, whereas the noise power increases as it is proportional to the bandwidth. Thus in mmWave transmissions the DL-to-UL inter-cell interference is not as strong and a global static UL/DL separation as in LTE is not required \cite{7499308,juanScheduling}. This means that 
% in the same frame one UE can transmit towards an AP (RN or BS) while, simultaneously, another UE receives from an AP, and 
the scheduler can exploit \textit{Dynamic Duplex} \cite{russellDynamic,juanScheduling}. In fact, the more recent 3GPP New Radio (NR) specification incorporates ``flexible'' slots that can be scheduled as either DL or UL \cite{3gpp38211}. Moreover, the IAB study item for future NR revisions has considered more than two hops \cite{8647977}. In \cite{russellDynamic} mmWave multi-hop scheduling with an arbitrary number of hops and dynamic duplex was studied. In \cite{russellDynamic} frequency-domain multiplexing was considered but not MU-MIMO, and the topology was still limited to a predefined tree where UEs cannot communicate with multiple APs and problems such as optimal tree formation, user attachment or routing were not addressed.

Moving away from the tree topology, in \cite{juanScheduling,Niu2017} multi-hop mmWave picocellular scheduling with a full mesh topology was considered, however in these works neither form of simultaneous transmission using MU-MIMO or frequency division was allowed. In an earlier work \cite{gomezITAoptimal} we considered the same problem introducing MU-MIMO for the receivers but not for the transmitter.

The present work provides a NUM analysis with MBP scheduling that extends \cite{juanScheduling,gomezITAoptimal}. MBP was shown to achieve network stability in single-hop ad hoc networks with fixed arrival rates and fixed link rates in \cite{Tass}. The NUM congestion control technique was introduced to maximize the user arrival rates in \cite{Kelly1997,Kelly1998}. The framework was generalized to multi-hop networks in \cite{Eryilmaz2007}. All these works assumed constant link rates, which is realistic in some cases such as sensor networks where nodes perform power control (reducing transmission power to the minimum necessary for a fixed link rate in order to improve battery life). NUM was extended to networks where links had a random transmit power in \cite{ModianoPower}, and QoS and delay were introduced in \cite{ZhouDelay2012}.

Unfortunately the existing body of work on NUM has predominantly applied to networks with a simple physical layer \cite{6615900}, e.g., considering single-antenna  omnidirectional radios incapable of decoding two simultaneous transmissions (interference-as-collision model). In order to apply the existing NUM know-how to cellular systems the above results need revision. In \cite{juanScheduling} NUM for mmWave multi-hop heterogeneous cellular networks is considered without MU-MIMO, leaving out one key technology in cellular modern physical layers. Nevertheless, \cite{juanScheduling} produced two analyses with and without interference in the mmWave multi-hop picocellular system, showing that the effect of interference was negligible. We verified their observations by simulating their algorithm, so in this paper we skip directly to a model assuming that interference is negligible in order to make the extension of the model in \cite{juanScheduling} to MU-MIMO tractable.

\subsection{Notation}
Calligraphic letters denote sets. $|\mathcal{A}|$ is the cardinality of $\mathcal{A}$. $\text{Int}\{\mathcal{A}\}$ is the interior region of $\mathcal{A}$. Script letters denote functions. $\dot{\mathscr{A}}^{-1}(x)$ is the inverse derivative of ${\mathscr{A}}(x)$. Bold uppercase and lowercase letters denote matrices and vectors, respectively. $\A^T$ is the transpose and $\A^H$ the Hermitian of $\A$. $\|\A\|_n=\left(\sum_{i,j}|a_{i,j}|^n\right)^{\frac{1}{n}}$ is the $\ell_n$ norm, where $\|\A\|=\|\A\|_2$. $\one_{N,M}$ is the $N\times M$ ``all-ones'' matrix. $\Ex{\A}{.}$ is the expectation with respect to the distribution of $\A$. $\textnormal{stack}(\ab_1\dots\ab_M)$ is the vector formed by stacking the $M$ vectors $\ab_1\dots\ab_M$ and $\text{Co}\{\A_1\dots\A_M\}$ is the set of all linear combinations of $\A_1\dots\A_M$.
%\end{spacing}

\section{System Model}
\label{sec:model}

%\begin{spacing}{1.56}
We represent a multi-hop wireless network as a directed graph $\mathcal{G}(\mathcal{N},\mathcal{L})$, where $\mathcal{N}$ is the set of all nodes (including BSs, RNs and UEs), and $\mathcal{L}$ is the set of all links. We also represent a set of differentiated \textit{traffic flows} in the network, $\mathcal{F}$. We denote the indices of the elements in each set by $n$, $\ell$ and $f$ respectively, and the cardinalities of the sets as $N$, $L$ and $F$. 

The set of links $\mathcal{L}$ contains one element per each pair of devices that can reach each other as transmitter and receiver. UEs can attach to multiple APs at the same time (RNs or BSs), RNs can communicate with other RNs, and BSs can communicate with any UE and RN in their range. Only UE-UE connections are forbidden. BSs differ from RNs in that they are connected to a wired backhaul. In general this model admits the definition of any arbitrary traffic flows. However, without loss of generality, we shall present our numerical simulations assuming there are 2 traffic flows per UE in the system: one for downlink data, with source at the BS and destination at the UE, and one for uplink data, with source at the UE and destination at the BS.

We assume that each node $n$ has knowledge of the set of neighbors connected to it, denoted by $\mathcal{A}(n)$, and of the corresponding channel coefficients. The maximum degree of the graph is $A_{\max}=\max_n|\mathcal{A}(n)|$. In mmWave, neighbor detection and channel estimation overhead is greatly reduced thanks to compressed sensing \cite{gomezcuba2019CSicc}. However, in order to fully exploit the antenna arrays at the UEs using MU-MIMO, we also assume a  ``multi-attachment'' mesh topology, that is, each UE and RN can be linked to more than one AP. This increases the estimation overhead because each node estimates more than one channel. The net effect on channel estimation overhead combining compressed sensing and multi-attachment is out of the scope of this paper.

We assume that each node uses the half-duplex MU-MIMO hybrid analog-digital mmWave physical layer architecture described in Appendix \ref{app:phy}. From the point of view of the network scheduling, the use of this MU-MIMO physical layer means that we assume that each node $n$ has $K(n)$ ``antenna ports.'' Half-duplex MU-MIMO means that each node is capable of transmitting $K(n)$ simultaneous signals to or receiving $K(n)$ signals from $K(n)$ neighbors at the same time. For RNs, these $K(n)$ simultaneous transmissions include the backhaul link with BS, which cannot be permanently in use and is scheduled with the same constraints as the links with UEs. We assume that $K(n)\geq |\mathcal{A}(n)|\;\forall n$ and thus transmission to or reception from all neighbors can be performed at once, but not both due to the half-duplex constraint. We leave the extension of our model to full duplex communications for future work. We also leave for future work the extension to the case $1<K(n)< |\mathcal{A}(n)|$, which would imply that MU-MIMO is possible but not all neighbors can be addressed at once. On the other hand, if $K(n)=1$ there is no MU-MIMO and each node can only communicate with one neighbor at a time (hereafter, ``one-to-one constraint'') as was assumed in most non-mmWave NUM literature \cite{Tass,Eryilmaz2007,Kelly1997,Kelly1998,ModianoPower} as well as in \cite{juanScheduling}.

%\end{spacing}

We assume that time is divided into frames with index $t$. All nodes are synchronized to the frame timing, and in each frame can change their configuration in terms of whether they transmit or receive, and, in the first case, how to split their total transmit power among the signals intended towards each neighbor. We denote the transmit/receive state of node $n$ with the boolean indicator $s_n(t)$, which is $s_n(t)=1$ if node $n$ transmits during frame $t$ and $0$ otherwise.
Denoting the total power budget at node $n$ by $P_n$, for each pair of nodes that form a link $\ell=(n,m)$, $n,m \in \mathcal{N}$ we define the normalized link power allocation $p_{n, m}(t)\in[0,1]$ subject to the constraint $\sum_{m\in\mathcal{A}(n)} p_{n, m}(t)\leq1$ to represent the power allocations of $n$ towards its neighbors. Note that due to the half duplex constraint we also must impose $ p_{n, m}(t)\leq s_n(t)(1-s_m(t))$.

We represent the state of all nodes in frame $t$ by the binary vector $\s(t)$ and denote the power allocations for all links by the vector $\pp(t)$ with $p_{n, m}(t)$ as the $(n-1)N+m$-th entry. The vector $\pp(t)$ is enough to fully identify the actions of all nodes in the network during frame $t$, and therefore the role of a scheduling policy should be to choose $\pp(t)$ for each frame $t$. However, the election of $\pp(t)$ is subject to the half duplex constraint, so scheduling constraints become more clear by writing $\s(t)$ explicitly. We denote the set of all possible power allocations assuming that $\s(t)$ is fixed as $\mathcal{P}(\s(t))$. This is a continuous convex set characterized by the sum-power constraint. The set of all state vectors $\s(t)$ is countable and contains all $2^N$ vectors of $N$ binary elements. Finally, the set of all possible schedules in the network is $\pp(t)\in\mathcal{P}={\displaystyle \bigcup_{\substack{\forall \s(t)} }} \mathcal{P}(\s(t))$. 

Using this notation, the rate of the transmission from $n$ to $m$ during frame $t$ is
\begin{equation} \label{eq:rate}
r_{n , m}(t) = \alpha_1  T_\mathrm{f}W  \log \left( 1 + \alpha_2  \frac{p_{n,m}(t)P_n|G_{n,m}|^2}{I_{n,m}(\pp(t))+W  N_o}  \right) 
\end{equation}
bits per frame, where $W$ is the bandwidth, $T_\mathrm{f}$ is the frame duration in seconds, $G_{n,m}$ is the equivalent complex gain between $n$ and $m$ of the MU-MIMO channel, $N_o$ is the thermal noise power spectral density, $I_{n,m}(\pp(t))$ is the interference power received in the link $(n,m)$ from other active transmitters, and $\alpha_1$ and $\alpha_2$ are the spectral and power efficiency penalties of the physical layer compared to the Shannon capacity. For illustration in our simulations we set $\alpha_2$ to $-3$ dB and $\alpha_1=1$.

\begin{figure*}[!b]
% ensure that we have normalsize text
\normalsize
% The spacer can be tweaked to stop underfull vboxes.
\vspace*{4pt}
% IEEE uses as a separator
\hrulefill
% Store the current equation number.
\setcounter{MYtempeqncnt}{\value{equation}}
% Set the equation number to one less than the one
% desired for the first equation here.
% The value here will have to changed if equations
% are added or removed prior to the place these
% equations are referenced in the main text.
\setcounter{equation}{1}
\begin{equation}
\label{eq:onequeue}
 q_n^f(t+1)=\begin{cases}
             0 & n \in \mathcal{D}_f\\
 \underset{\textrm{previous}}{\underbrace{q_n^f(t)}}+\underset{\textrm{incoming}}{\underbrace{{\displaystyle\sum_{m\in\mathcal{A}(n)}} r_{m,n}^f(t)}}-\underset{\textrm{outgoing}}{\underbrace{{\displaystyle\sum_{m\in\mathcal{A}(n)}}r_{n,m}^f(t)}}+\underset{\textrm{exogenous}}{\underbrace{a_s^f(t)}}& n\notin \mathcal{D}_f\\
            \end{cases}
\end{equation}
% Restore the current equation number.
\setcounter{equation}{\value{MYtempeqncnt}}
\end{figure*}

The complex channel gain is defined as $G_{n,m}\triangleq g_{n, m}(\w_{n, m}^{\mathrm{r}})^H\Hb_{n, m}\w_{n, m}^{\mathrm{t}}$, where $g_{n, m}$ is the pathloss, $\Hb_{n,m}$ is the channel matrix, and $\w_{n, m}^{\mathrm{r}}$ and $\w_{n, m}^{\mathrm{t}}$ are the receiver and transmitter beamforming vectors. The exact model for $G_{n,m}$, as well as $I_{n,m}(\pp(t))$, is given in Appendix \ref{app:phy}.

The expression \eqref{eq:rate} can model rate in any arbitrary multi-hop MU-MIMO network, but in its evaluation for a mmWave multi-hop picocellular system we adopt the assumption that interference is much weaker than noise,  $I_{n,m}(\pp(t))\ll W  N_o\;\forall\pp(t)\in\mathcal{P}$. Thus the denominator in \eqref{eq:rate} simplifies to $\simeq W  N_o$ and $r_{n , m}(t)$ only depends on $p_{n,m}(t)$. We give in Appendix \ref{app:phy} a full review of the mmWave physical layer framework that permits to disregard the interference  because of the high pathloss and antenna directivity. The distinction between interference-limited and noise-limited regimes in single-hop mmWave cellular systems is discussed in \cite{7499308}, and the achievable rates were shown to display negligible differences between a so-called Actual Interference and a so-called Interference-Free mmWave network models in \cite{juanScheduling}.

Even with non-negligible interference, after we choose the schedule vectors $(\s(t),\pp(t))$, this fixes the values of the link rates in \eqref{eq:rate}. These rates characterize the evolution of traffic and queues in the network from frame $t$ to frame $t+1$. As said above, we assume there are $F$ flows that can be arbitrarily specified. We assume that at each node $n$ a separate queue can be maintained for each flow $f$, and we denote the number of bits in each queue by $q_{n}^f$. We denote with the vectors $\q_n$, $\q^f$ and $\q\triangleq \textnormal{stack}(\q^1 \dots \q^F)$ the queue lengths of all flows at node $n$, the queue lengths dedicated to flow $f$ at all nodes, and all the queues of the network, respectively.

We assume that for each flow $f\in\mathcal{F}$ some nodes, denoted by sets $\mathcal{S}_f$ and $\mathcal{D}_f$, are \textit{sources} and \textit{destinations} of information, respectively. Each flow can have more than one source or destination. Each source of flow $f$ during frame $t$ generates a random number of bits denoted by $a_s^f(t)$. We call these ``exogenous arrivals'' because their generation at the source models their arrival from ``outside the network,'' presumably from an application. When $|\mathcal{S}_f|>1$, ``multi-source'' traffic of the same flow is supported.  Each destination of a flow withdraws from the network all bits that reach it, always forcing its queue to $0$. Thus, when $|\mathcal{D}_f|>1$, ``anycast'' traffic is supported, not ``multicast''\footnote{Anycast traffic must reach only one of the destinations in a set while multicast must reach all destinations in the set.}. We assume that the average exogeneous arrivals rates at each source are elastic, i.e.,
\begin{definition}
The \textbf{elastic} traffic arrival process associated with flow $f$ in source node $s\in \mathcal{S}_f$ is a stochastic process with a time-varying mean arrival rate injected into the network $\lambda_s^f(t)=\Ex{}{a_s^f(t)}$. The long-term mean arrival rate of the source is defined as $x_s^f = {\displaystyle \lim_{T \to \infty} }\frac{1}{T} \sum_{t = 1}^{T} \lambda_s^f (t)$.
\end{definition}
We denote vectors $\ab(t),\lambdab(t),\x$ as the stacked packet arrival realizations, time-varying mean, and long-term average of the packet arrival processes, respectively. 

In addition to transmitting at the same time to multiple neighbors with the rates \eqref{eq:rate}, each transmitter has the ability to choose which queue (i.e., which flow) is serviced using the rates of each link. We denote by $r_{n,m}^f(t)$ the rate of link $n,m$ that node $n$ dedicates to serving the queue $q_n^f(t)$. The total link rates must not be exceeded $\sum_{f\in\mathcal{F}} r_{n,m}^f(t)\leq r_{n,m}(t)$ and the node cannot send more bits than there are in its queue  $\sum_{m\in\mathcal{A}(n)} r_{n,m}^f(t)\leq q_n^f\;\forall f$.

Considering together the exogeneous arrivals, outgoing links, incoming links, and discarding at destination, from each frame to the next each queue in the network evolves according to \eqref{eq:onequeue} at the bottom of the previous page,
\stepcounter{equation}%this is in place of the float which uses manual numbering
where a major difference between our MU-MIMO multi-hop network model and prior NUM literature is that two or more terms in the sum ${\displaystyle\sum_{m\in\mathcal{A}(n)}} r_{m,n}^f(t)$ can be non-zero at the same time, and likewise for ${\displaystyle\sum_{m\in\mathcal{A}(n)}} r_{n,m}^f(t)$. However, at least one of the two sums must always be zero due to the half-duplex constraint.
We adopt a compact matrix notation defining $\D$ as the diagonal matrix with $D_{n+N(f-1),n+N(f-1)}=1$ if $n\notin\mathcal{D}^f$ and zero elsewhere. Thus, \eqref{eq:onequeue} for all queues in $\q$ is
\begin{equation}
 \label{eq:qupdate}
 \q(t+1)=\q(t)+\D\left[(\R^T(t)-\R(t))\one_{NF,1}+\ab(t)\right]
\end{equation}
where the link rate matrix $\R(t)$ is defined with coefficients $R_{n+N(f-1),m+N(f-1)}=r_{n,m}^{f}(t)$.

An example illustration of our queue model is given in Fig. \ref{fig:zoom}. In this example we show the details for one link in the network graph, i.e.,  the one between nodes $n$ and $m$. Since $s_n(t)=1$ and $s_m(t)=0$, in frame $t$ node $n$ acts as a transmitter and $m$ and a receiver. $n$ allocates power $p_{n,m}$ to this transmission, achieving rate $r_{n,m}$, and the reverse link $r_{m,n}$ is not active. There are three flows in the example. Each node maintains one separate queue for each flow. Flows 1 and 3 have source and destination somewhere else in the rest of the network, represented as black dots, whereas node $n$ is the source of flow 2 and generates packets in its queue with rate $\lambda_n^{(2)}(t)$.

\begin{figure}
 \centering
 \includegraphics[width=.7\columnwidth]{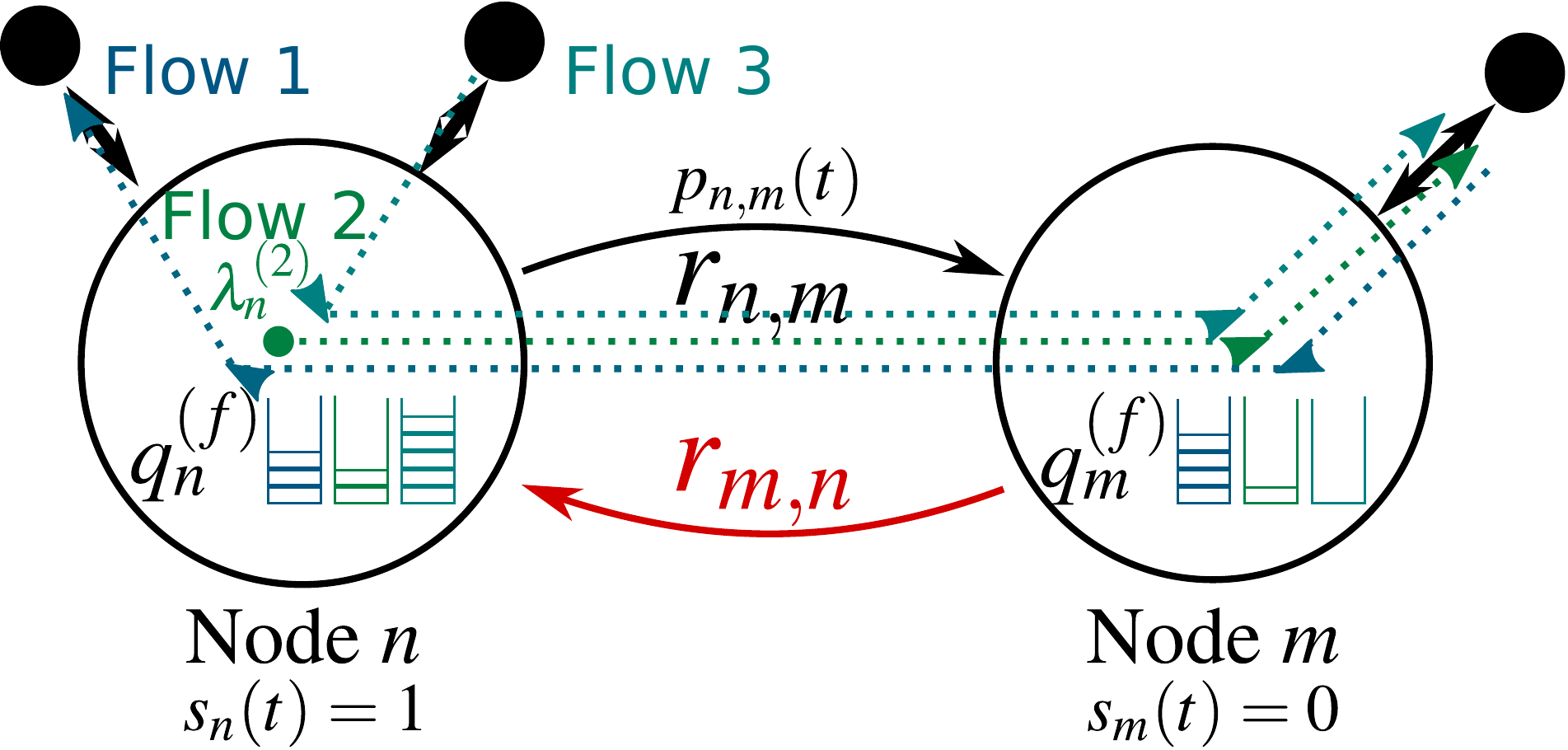}
 \caption{Zoom on one link in the network graph and the mathematical entities involved in forwarding}
 \label{fig:zoom}
\end{figure}

\section{General MU-MIMO NUM}
\label{sec:optimization}
The Network Utility Maximization optimizes the network performance measured by a \textit{utility function}. For each flow we define a continuous non-decreasing function that assigns a value $\mathscr{U}^f(\sum_{\mathcal{D}_f}y_n^f)$ to the successful delivery of a rate of $\sum_{\mathcal{D}_f}y_n^f$ bits per second to its destinations. A queue is stable if it does not grow unbounded, i.e., $\lim_{T\to\infty} \frac{1}{T}\sum_{t=1}^T q_n^f(t)<\infty$ with probability $1$. The network is stable when all its queues are stable, i.e., $\lim_{T\to\infty} \frac{1}{T}\sum_{t=1}^T\|\q(t)\|_1<\infty$ w.p.$1$.
\begin{definition}
 The \textbf{throughput capacity region} $\x\in\mathcal{X}$ is the set of long-term average rate vectors at the sources for which there exists a scheduling policy such that the network is stable.
\end{definition}
For any $\x\notin\mathcal{X}$ a non-zero number of bits of the arrival rate $\x$ stall in the queues for an infinite time and never reach the destinations. Conversely, if $\x\in\mathcal{X}$,  long term average rates at the sources and the long term average throughput that arrives at the destination are identical, $\sum_{\mathcal{D}_f}y_n^f=\sum_{\mathcal{S}_f}x_n^f$, even if the instantaneous values during each frame may be different. Relying on this conservation of traffic, we define the \textbf{Network Utility Maximization problem} as follows
\begin{equation}
\label{eq:NUMproblem}
  \max_{ \x\in \mathcal{X}} \sum_{f=1}^{F} \mathscr{U}^f \left(\sum_{n=1}^{N} x_n^f\right).
\end{equation}

Two examples of classic utility metrics are the linear and logarithmic functions. Linear utility results in the maximization of the sum rate. The function $\mathscr{U}(y)=\frac{1}{2}\log(y)$ results in the well-known ``proportional fair'' rate allocation \cite{Kushner2004}. Other utility functions can be designed, for example assigning different functions to different flows for service differentiation.

The NUM problem \eqref{eq:NUMproblem} presents two issues: first, it relies on conservation of traffic so that when the arrival rates at the sources are $\x$ the throughput at the destinations is the same. In other words, to achieve the network utility we need to find the scheduler that ensures network stability when the long term exogenous arrival rates at the sources are $\x^*$ where $\x^*$ is the solution to \eqref{eq:NUMproblem}. Secondly, the region $\mathcal{X}$ is generally unknown and choosing $\x^*$ with fixed rate inelastic traffic is difficult. In order to achieve both stability and a solution to the NUM problem without a priori knowledge of $\mathcal{X}$, we look for a scheduling technique that \textit{always} guarantees network stability $\forall \x\in\mathcal{X}$ and, separately, we design a rate-adaptation technique for elastic traffic that modifies $\lambdab(t)$ such that, \textit{at run-time}, the long term average rates in the network converge to $\x^*$.

\begin{definition}
A scheduling policy is \textbf{throughput optimal} if it makes the network stable $\forall\x\in\mathcal{X}$. 
\end{definition}

In the throughput optimality and NUM literature it is well known that Maximum Back Pressure is throughput optimal. We present a MU-MIMO MBP generalization in Algorithm \ref{alg:mmbp}. In our MU-MIMO extension we introduced the factor $\xi_{n,m}$ in line 3, which is applied in line 4, to address the constraint $\sum_{m\in\mathcal{A}(n)} r_{n,m}^f(t)\leq q_n^f\;\forall f$ when the transmissions to several receivers are dedicated to the same flow queue. First, line 2 selects for each link $(n,m)$ the flow with the highest ``queue pressure.'' Second, line 3 divides the bits in each queue $q_{n}^f$ across all the links that selected the same flow. Third, the actual rate for each link is established in line 4 as the minimum between the rate of the link and the number of available bits in the queue. Fourth and finally, the state vector $\s(t)$ and power allocation vector $\pp(t)$ are selected to maximize the sum queue-pressure metric $r_{n,m}^f(q_n^f-q_m^f)$ over all the links in the network.

\begin{theorem} 
\label{the:MBP}
The \textit{MU-MIMO MBP} scheduling algorithm in Alg \ref{alg:mmbp} is throughput optimal.

\end{theorem}
\begin{proof}

The proof is given in Appendix \ref{app:PAC}, which is a variation of the proof with random link rates in \cite{Eryilmaz2007}. We define the random ``network state'' variable $\uu(t)=(\q(t),\R(t))$ which includes both the queue states and the rate matrix chosen by MBP. By definition the Markov chain $\uu(t)\to\uu(t+1)$ is \textit{irreducible} (any state is reachable form any other). We show that if $\x\in\mathcal{X}$ then the Markov Chain is \textit{positive recurrent}, i.e. it takes on average a finite number of transitions to return to a set of states that are associated with bounded queue lengths ($\|\q(t)\|^2<C_1$).
\end{proof}

\begin{remark}
 In Algorithm \ref{alg:mmbp} all the coefficients in $\pp(t)$ may affect the link rate $r_{n,m}(\pp(t))$ and we do not need to assume that interference is negligible in \eqref{eq:rate}. Theorem \ref{the:MBP} applies and MU-MIMO MBP as in  Alg. \ref{alg:mmbp} is a throughput optimal scheduler for any MU-MIMO multihop network model.
\end{remark}

\begin{algorithm}
\small
\caption{MU-MIMO Maximum Back Pressure}
\label{alg:mmbp}
\begin{algorithmic}[1]
\FORALL {$t$}
  \STATE {$f_{n,m}^*=\arg\max_{f}(q_{n}^f(t)-q_{m}^f(t))$}
  \STATE {$\xi_{n,m}=\frac{r_{n,m}(\pp(t))}{\displaystyle\sum_{m':f_{n,m}^*=f_{n,m'}^*} r_{n,m'}(\pp(t))}$}  
  \STATE {$r_{n,m}^f(\pp(t))=\begin{cases}
		  \min(r_{n,m}(\pp(t)),q_{n}^f(t)\xi_{n,m})& f=f^*_{n,m}\\
		  0& \text{otherwise}
		\end{cases}$}
  \STATE {$(\s(t),\pp(t))=$
  \begin{equation}
   \arg \max_{\substack{ \s(t)\in\{0,1\}^N \\ \pp(t)\in\mathcal{P}(\s(t)) \\ \textnormal{s.t. lines 2-4}}}\sum_{n=1}^N\sum_{m=1}^N\max_f r_{n,m}^f(\pp(t))(q_{n}^f-q_{m}^f)\label{eq:MBP}
  \end{equation}}
\ENDFOR
\end{algorithmic}
\end{algorithm}

Next, we design our rate-adaptation technique inspired by a multi-objective optimization that associates  NUM problem \eqref{eq:NUMproblem} and the queue length via an auxiliary scalar constant $V$.
\begin{equation}
\label{eq:EPSproblem}
 \x^{V}=\arg\max_{\x} V \sum_{n,f}\mathscr{U}(x_{n}^{f})-\Ex{\q}{\q^T\x}.
\end{equation}

The Adaptive NUM Congestion Control (Algorithm \ref{alg:anumcc}) is based on the gradient of \eqref{eq:EPSproblem} equal to zero. When $\mathscr{U}^f()$ is strictly concave, since by definition $\mathscr{U}^f()$ is non-decreasing, the derivative is monotonically decreasing with rate. Thus the inverse derivative in ANCC increases the rate when the queue is small, and vice-versa. The scalar constant $V$ scales the response to queue length and the proximity to the optimal solution. Higher values of $V$ make $\|\x^V-\x^*\|^2$ smaller, where $\x^*$ is the true NUM solution. However, increasing $V$ means that ANCC reduces the rate less when queues grow, bringing the network closer to instability and making its network-state Markov process need many frames to converge to the steady-state distribution. 

\begin{algorithm}
\small
\caption{Adaptive NUM CC}
\label{alg:anumcc}
\begin{algorithmic}[1]
\FORALL {$t$}
\STATE{$R_{\max}=\max_{n,m}(r_{n,m|p_{n,m}=1})$}
\STATE{$A_{\max}=\max_{n}(\mathcal{A}(n))$}
\STATE{$\lambda_{\max}=A_{\max}R_{\max}$}
\STATE{$\displaystyle \lambda_n^f(t)=\begin{cases}
                        \max(\min(\dot{\mathscr{U}}^{-1}(\frac{q_n^f(t)}{V},\lambda_{\max}),0) & n\in \mathcal{S}_f\\
                        0& \text{otherwise}
                       \end{cases} \inlineeqnum\label{eq:cc}$}
\ENDFOR
\end{algorithmic}
\end{algorithm}

\begin{theorem} 
\label{the:NUMCC}
In a network with MBP scheduling and $\lambdab(t)$ controlled by the Adaptive NUM CC algorithm (Alg \ref{alg:anumcc}) with a strictly concave utility function $\mathscr{U}()$, long-term rates converge to $\x^{V}$, the solution of \eqref{eq:EPSproblem}. Moreover $\x^V$ is arbitrarily close to $\x^*$ as $V\to\infty$.
\end{theorem}
\begin{proof}
The proof of convergence to $\x^V$, given in Appendix \ref{app:PAC}, is an extension of the proof of Theorem \ref{the:MBP}, and is a variation of the proof in \cite{Eryilmaz2007}. We define the random variable $\uu(t)=(\q(t),\R(t))$ and show that  $\uu(t)\to\uu(t+1)$ is a positive recurrent irreducible Markov process where the system returns in average finite time to states that satisfy 1) that the queues are bounded by a constant $\|\q(t)\|^2<C_1V$ and 2) that the long-term averages at the ANCC sources satisfy $\|\x-\x^V\|<C_2/V$. Thus the network is stable and we can make the long-term average rates as close to the solution of \eqref{eq:EPSproblem} as we desire by increasing $V$. Finally ${\displaystyle \lim_{V\to\infty}}\x^V=\x^*$ is a property of the multi-objective optimization \eqref{eq:EPSproblem}.
\end{proof}

\section{MU-MIMO MBP Scheduler for mmWave}
\label{sec:mmWaveImplementation}

Theorems \ref{the:MBP}, \ref{the:NUMCC} apply to any multi-hop MU-MIMO network. Nevertheless, the difficulty in addressing interference in MU-MIMO multip-hop networks is not the theoretical proof of MBP NUM optimality, but the implementation of line 5 of Algorithm \ref{alg:mmbp} as an optimization. 

Thanks to the dual notation $(\s(t),\pp(t))$, a separation in two subproblems can be written:
\begin{enumerate}[i)]
 \item Selection of the optimal power allocation as a function of a given node-states vector, 
 \begin{equation}
 \begin{split}
 \label{eq:pallocsubproblem}
  \pp(t)^*&=\rho(\s(t))\\
  &\triangleq \max_{\substack{ \pp(t)\in\mathcal{P}(\s(t)) \\ \textnormal{s.t. fixed }\s(t)}}\sum_{n=1}^N\sum_{m=1}^N\max_f r_{n,m}^f(\pp(t))(q_{n}^f-q_{m}^f) ,
  \end{split}  
 \end{equation}
which is a constrained problem over the domain of non-negative real numbers, and
 \item The selection of $\s(t)$, an unconstrained binary problem assuming \eqref{eq:pallocsubproblem} is known for each $\s(t)$
 \begin{equation}
 \label{eq:MBPDBSGsubproblem}\max_{\s(t)\in\{0,1\}^N} \sum_{n=1}^N\sum_{m=1}^N\max_f r_{n,m}^f(\rho(\s(t)))(q_{n}^f-q_{m}^f).
 \end{equation}
\end{enumerate}

For the first problem, in mmWave we can adopt the assumption that interference is negligible and write $I_{n,m}(\pp(t))+WN_o\simeq WN_o$ in \eqref{eq:rate}. With this, power allocation can be solved locally for each node $n$. For the second problem, only links $(n,m)$ where $s_n(t)=1$ and $s_n(t)=0$ communicate. Thus we say that the state vector $\s(t)$ induces a \textbf{Directed Bipartite SubGraph} (DBSG) of $\mathcal{G}(\mathcal{N},\mathcal{L})$, denoted by $\mathcal{SG}(\mathcal{N}_T\bigcup \mathcal{N}_R,\mathcal{L}_{TR})$. In the DBSG the set of edges $\mathcal{N}$ is divided into two disjoint sets $\mathcal{N}_T\triangleq\{n\in\mathcal{N}:s_n(t)=1\}$ and $\mathcal{N}_R\triangleq\{m\in\mathcal{N}:s_m(t)=0\}$, and the set of active links satisfies $\mathcal{L}_{TR}\triangleq\{(n,m)\in\mathcal{L}:s_n(t)=1,s_m(t)=0\}$. We say the second problem is the Maximum Back Pressure DBSG of the network subject to a power allocation $\rho(\s(t))$.

Unfortunately, even without interference, the implementation is still challenging because $\rho(\s(t))$ changes with $\s(t)$. To make the problem tractable we introduce a simplified fixed power allocation that leads to fixed link weights, making the problem a Maximum Weighted DBSG (MWDBSG).

We rely on the problem separation to highlight that Theorems \ref{the:MBP} and \ref{the:NUMCC}, and their proofs, can be applied to MU-MIMO networks with any additional power-allocation rule:
\begin{lemma}
\label{lem:constrained}
A MU-MIMO MBP + ANCC with power allocation constraints is still NUM optimal. 
\end{lemma}
\begin{proof}
 Denote by $\breve{\mathcal{P}}(\s(t))\subset\mathcal{P}(\s(t))$ a set of valid power allocations given $\s(t)$ under any additional constraint we desire. Substituting $\mathcal{P}(\s(t))$ by $\breve{\mathcal{P}}(\s(t))$ in Algorithm \ref{alg:mmbp}, denote the new power allocation subproblem under constraints as 
 \begin{equation}
 \label{eq:pconstsubproblem}\breve{\rho}(\s(t))\triangleq \max_{\substack{\pp(t)\in\breve{\mathcal{P}}(\s(t))\\ \textnormal{s.t. fixed }\s(t)}}\sum_{n=1}^N\sum_{m=1}^N\max_f r_{n,m}^f(t)(q_{n}^f-q_{m}^f).
 \end{equation} Define the constrained MU-MIMO MBP scheduler as the one that chooses $\s(t)$ as the maximum back-pressure DBSG when the power allocation is $\breve{\rho}(\s(t))$. Defining the constrained throughput capacity region $\breve{\mathcal{X}}\subset\mathcal{X}$ as the set of all arrival rate vectors that can be stabilized under the additional constraints, the proofs of Theorems \ref{the:MBP} and \ref{the:NUMCC} in Appendix \ref{app:PAC} remain valid substituting $\x\in\mathcal{X}$ by $\x\in\breve{\mathcal{X}}$. Thus this constrained MU-MIMO MBP with ANCC is optimal with regard to a ``constrained version'' of the NUM problem substituting $\x\in\mathcal{X}$ with $\x\in\breve{\mathcal{X}}$ in \eqref{eq:NUMproblem}.
\end{proof}

This modified version of Theorems \ref{the:MBP} and \ref{the:NUMCC} allows us to invoke the NUM optimality result for any power allocation strategy $\breve{\rho}(\s(t))$ of our choice.
In particular, we can discuss three power allocation techniques that make practical sense in multi-hop mmWave picocellular networks:

 \subsubsection{Independent Amplifiers with Fixed Power (proposed in this paper)}
 In this paper we present a numerical solver for the optimal scheduler with the additional assumption that power allocation must be equal. mmWave radio hardware is often built using hybrid analog-digital circuit designs \cite{Orhan2015,Mo2016,Abbas2016}. Thus, we can assume that each MU-MIMO simultaneous transmission makes use of a separate power amplifier in the analog circuitry, and the power allocation takes fixed equal values as $p_{n, m}(t)=s_n(1-s_m)P_n/|\mathcal{A}(n)|\;\forall m\in\mathcal{A}(n)$. Moreover, we maintain the assumption that interference is negligible \cite{7499308,juanScheduling}. Combined, these two assumptions allow to represent each link rate as a fixed constant $R_{n,m}=r_{n,m}(1/|\mathcal{A}(n)|)$, where we evaluate \eqref{eq:rate} without interference and with power allocation $p_{n,m}(t)=1/|\mathcal{A}(n)|$. For each link, the queue-pressure weight is defined as $w_{n,m}=\min(R_{n,m},\xi_{n,m}q_{n}^f)\max_f(q_n^f-q_m^f)$, and the MBP scheduler \eqref{eq:MBP} under these constraints is the Maximum Weighted Directed Bipartite SubGraph (MWDBSG) of $\mathcal{G}(\mathcal{N},\mathcal{L})$,
 \begin{equation}
 \label{eq:mwdbsg}\max_{\mathcal{SG}(\mathcal{N}_T\bigcup \mathcal{N}_R,\mathcal{L}_{TR})\subset \mathcal{G}(\mathcal{N},\mathcal{L})}\sum_{(n,m)\in\mathcal{L}_{TR}} w_{n,m}. \end{equation}

This graph partition is more difficult than the MWM implementation of MBP multi-hop network scheduling without MU-MIMO used in prior works \cite{Tass,Eryilmaz2007,Kelly1997,Kelly1998,ModianoPower,juanScheduling}. In order to solve the MWDBSG problem we use a MILP toolbox. We define the binary indicator $b_{n,m}(t)\in\{0,1\}$ to indicate if the link $(n,m)$ is chosen, that is $s_n(t)=1$ and $s_m(t)=0$. We define a stack vector with these binary variables as $\bb$. The scheduling variables for the general case $\s(t)$ and $\pp(t)$ can be reconstructed using $p_{n,m}(t)=b_{n,m}(t)/\mathcal{A}(n)$ and $s_n(t)=\min(1,\sum_{m\in\mathcal{A}(n)} b_{n,m}(t))$. We can use the linear expression $b_{n,m}+\frac{1}{|\mathcal{A}(n)|}\sum_{m'\in\mathcal{A}(n)}b_{m',n}\leq 1$ as a replacement of the half duplex scheduling constraint. Since $b_{n,m}$ takes binary values, this linear constraint allows either $b_{n,m}=1$ or $\frac{1}{|\mathcal{A}(n)|}\sum_{m'\in\mathcal{A}(n)}b_{m',n}>0$, but not both. Therefore, the MWDBSG optimization, which is \eqref{eq:MBP} constrained to equal power allocation, can be rewritten as a MILP problem. Finally, standard MILP toolboxes can solve the following problem and implement the scheduler:

\begin{equation}
\label{eq:milproblem}
\begin{split}
 &\max_{\bb} \sum_{n} \sum_{m\in \mathcal{A}(n)} b_{n,m}w_{n,m}
  \\&
  \textnormal{ s.t. }b_{n,m}+\frac{1}{|\mathcal{A}(n)|}\sum_{m'\in\mathcal{A}(n)}b_{m',n}\leq 1 \;\forall n,m\\
 \end{split}
\end{equation}

\subsubsection{Destination Selection (prior work)}
In our earlier work \cite{gomezITAoptimal} we considered MU-MIMO at the receiver but not at the transmitter. This assumption, hereafter ``K-to-one,'' is realistic if each device's radio has multiple phased-array blocks but only one power amplifier, so MU-MIMO can only be used at the receiver side. Power allocation in this context reduces to destination-selection at the transmitter. In our prior work \cite{gomezITAoptimal} we designed a heuristic Message-Passing algorithm that performed close to the optimal. In this paper, instead, we obtain the optimal MBP K-to-one scheduler using a constrained variant of the MILP algorithm \eqref{eq:milproblem}. For this we add the constraint $\sum_{m}b_{n,m}\leq 1\;\forall n$ and change the power allocation to $p_{n,m}(t)=b_{n,m}(t)$.

 \subsubsection{Waterfilling (future extension)}
 A potential generalization of the model in this paper would be to allow dynamic power allocations that adapt to the state of the neighbors.
%  , not pre-reserving power towards those that are in the ``transmitting'' state and will not detect any signal.
 Many state-of-the-art physical layer proposals for MU-MIMO single-hop cellular systems maximize the sum rate over all receivers. Maintaining the assumption that the interference is negligible  in mmWave, but enabling dynamic power allocation, each transmitter can independently maximize its individual sum rate towards its set of neighbors in a receiving state. This would be achieved with the water-filling algorithm, choosing $p_{n, m}(t)=s_n(t)(1-s_m(t))\max\left(0,\frac{1}{\lambda^*}-\frac{WN_o}{P_n|G_{n,m}|^2}\right)$ where $\lambda^*$ is the ``water level'' associated with the Lagrange multiplier for the constraint $\sum_{m\in\mathcal{A}(n)}p_{n,m}(t)=1$. Unfortunately, in this scenario $p_{n, m}(t)$ varies with $\{s_m(t),m\in\mathcal{A}(n)\}$, and an efficient algorithm to find the optimal $\s(t)$ remains unknown. This scenario is important because waterfilling power allocation is common in existing cellular systems, nonetheless we leave it for future work. 
 
\section{Numerical Results}
\label{sec:numeric}

We simulate randomly generated mmWave heterogeneous picocell networks with RNs for coverage extension as illustrated in Fig. \ref{fig:picocell}. We assume that BSs are located with an Inter Station Distance of $346$ m despite the fact that the range of a single BS is approximately $120$ m, and the coverage extended to $200$ m with RNs. A similar picocell network model has been considered for mmWave cellular capacity evaluations in \cite{Akdeniz2014,juanScheduling}. The channels and physical layers are modeled using the NYU mmWave model discussed in Appendix \ref{app:phy}. Since we assume that the interference is negligible, it is sufficient that we simulate the cell formed by a single BS with the RNs and UEs attached to it. That is, we only simulate the devices contained in the green hexagon in Fig. \ref{fig:picocell}. There is no need to simulate the ``encircling'' neighboring cells and the nodes in them to reproduce realistic interference since their power is negligible, unlike prior work for 4G LTE capacity evaluations where simulation of neighboring cells was imperative \cite{fgomez2014improvedrelaying}.

We assume 10 randomly located UEs uniformly distributed in the disk with radius $200$ m, a BS at the center of the disc, and four IAB RNs at $115$ m from the BS with a $90^o$ rotation. We repeated the simulation over 50 ``drops,'' i.e., realizations of the random UE location process. Two nodes are declared ``connected'' if they meet a maximum omnidirectional (i.e., without beamforming) pathloss threshold of $200$ dB, as in the neighbor-detection model in \cite{Barati2015}. All links that satisfy the minimum pathloss threshold and are not UE-UE form the network mesh topology. A few examples of simulated random network graphs are illustrated in Fig. \ref{fig:topo}. Each simulation ran for $10^5$ frames to let the network state Markov process approach its steady-state distribution. Due to the high frequency and bandwidth, we assumed a mmWave frame duration of $T_f=1$ $\mu$s. Thus $10^5$ frames span only $100$ ms of network operation and modeling the network as quasi-static is reasonable. The radio hardware parameters are given in Table~\ref{tab:param_tab}. 

\begin{figure}
  \centering
  \subfigure[Simulation no. 3]{
    \includegraphics[width=.6\columnwidth]{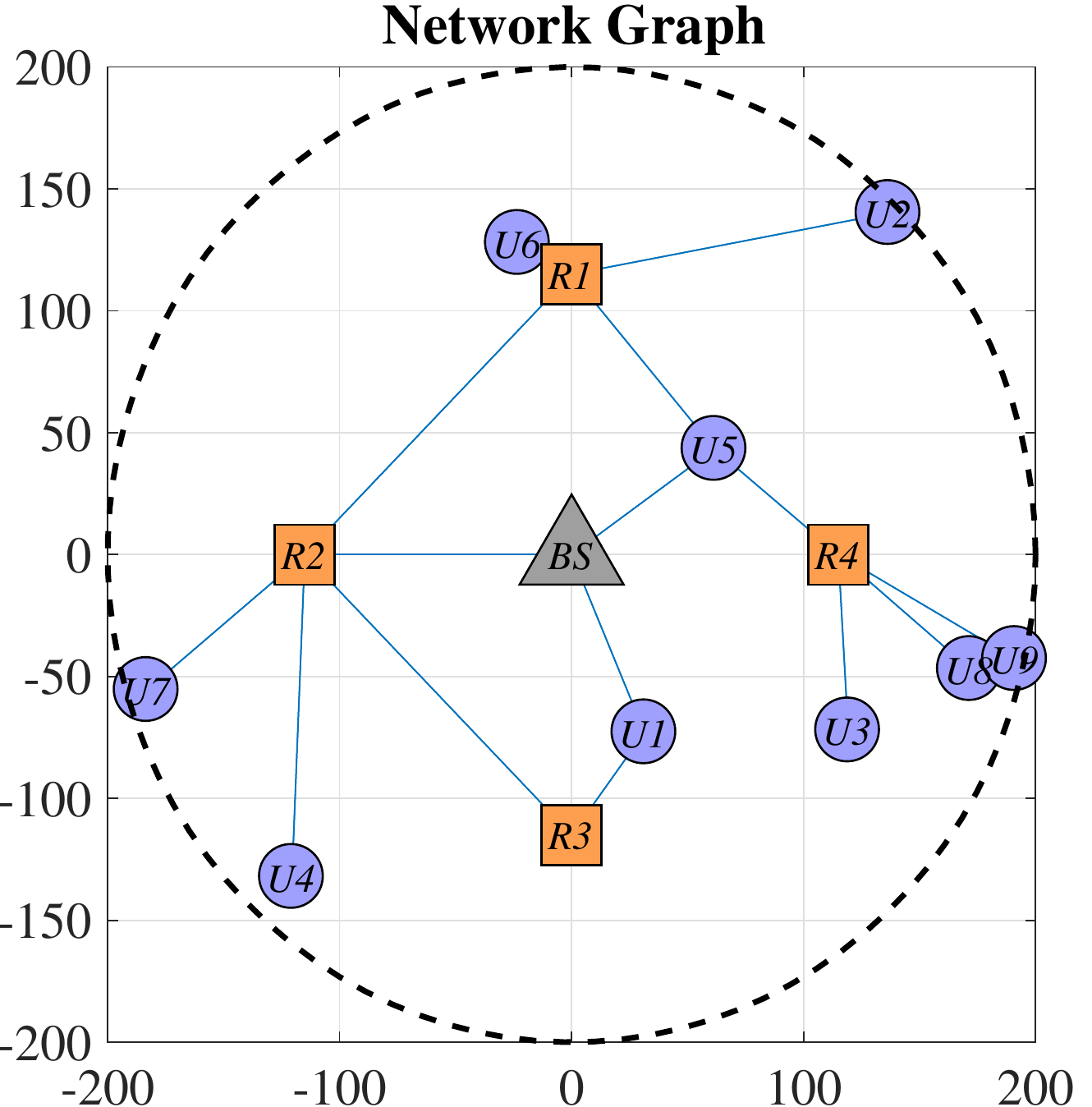}
    }
 \hspace{.2in}
  \subfigure[Simulation no. 7]{
    \includegraphics[width=.6\columnwidth]{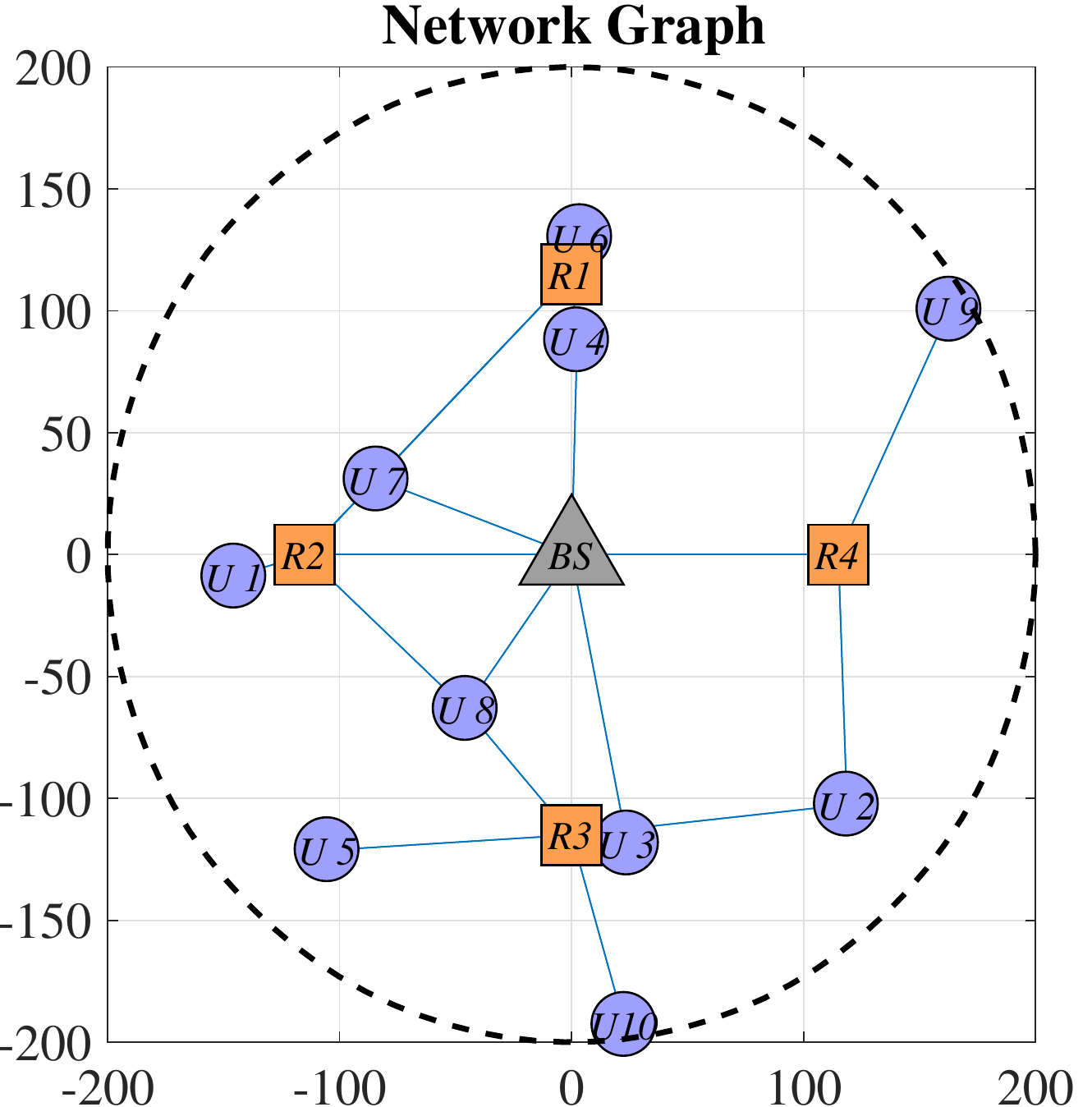}
    }
 \hspace{.2in}
  \subfigure[Simulation no. 20]{
    \includegraphics[width=.6\columnwidth]{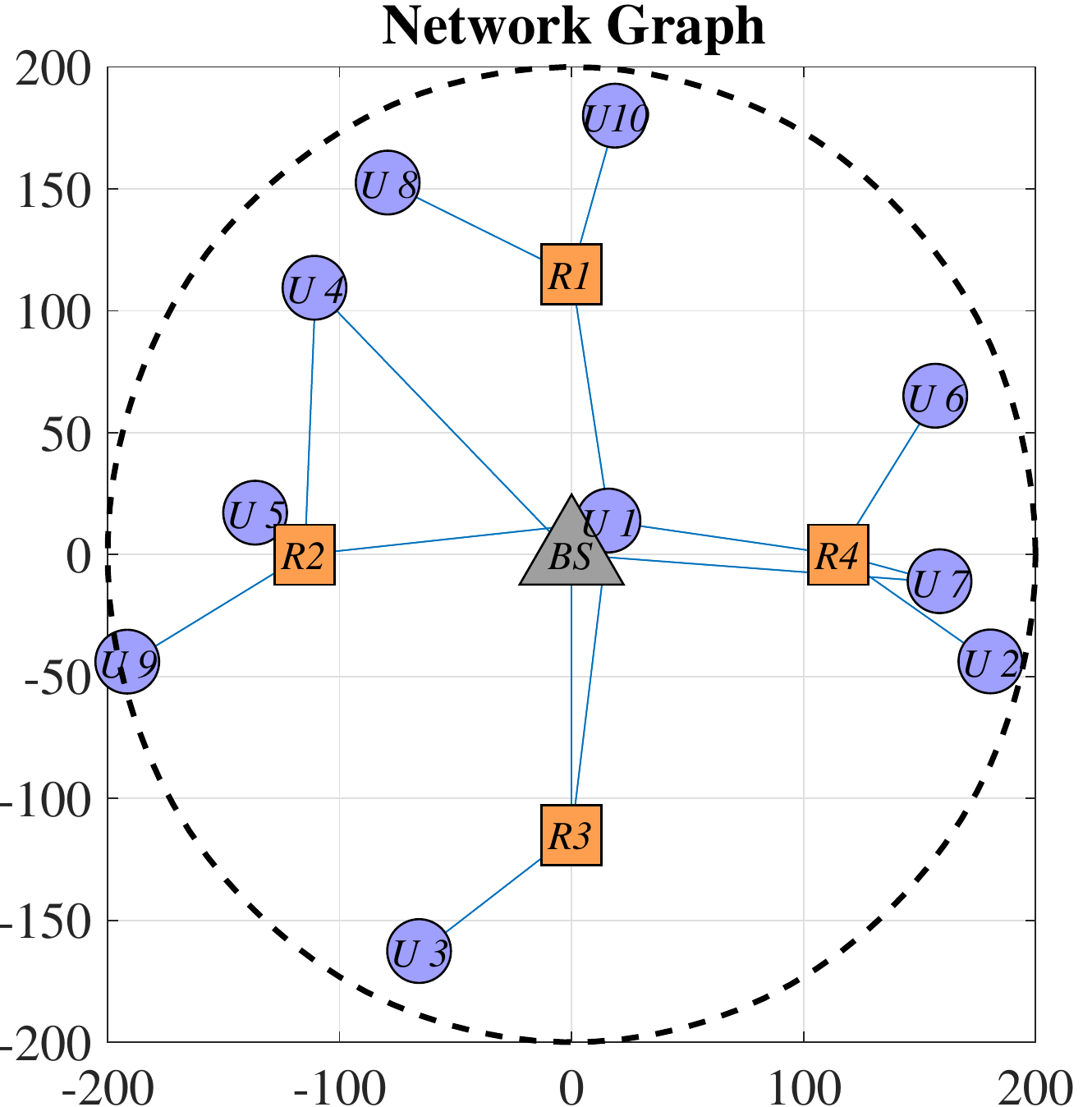}
    }
  \caption{Examples of simulated mmWave multi-hop heterogeneous picocells. A total of 50 drops were simulated.}
  \label{fig:topo}
\end{figure}

We assume two traffic flows are demanded by each UE: one uplink with source at the UE and destination at the BS, and one downlink with source at the BS and destination at the UE. All source rates apply ANCC (Alg. \ref{alg:anumcc}) with utility function $\mathscr{U}^f(x)=\frac{1}{2}\log(x)$ and the congestion control tuning parameter was set to $V=10 R_{\max}^2$ where $R_{\max}$ is the maximum link rate.

\begin{table}[!t]
  \centering
   \caption{mmWave Physical Layer Simulation Parameters}
   \label{tab:param_tab}
\begin{tabular}{r|l|l|l}
% \textbf{Parameter} & \multicolumn{3}{c}{\textbf{Values}}                            \\ \hline
\textbf{Node Type} & \textbf{BSs} & \textbf{RNs} & \textbf{UEs}                                         \\ \hline	
Carrier Frequency  & \multicolumn{3}{c}{28 GHz}                                                     \\\hline
System Bandwidth   & \multicolumn{3}{c}{1 GHz}                                                       \\\hline
Transmission Power & 30 dBm & 25 dBm & 20 dBm                        \\\hline
Receiver Noise Fig.       & 5 dB & 6 dB & 7 dB                               \\\hline
Planar ant. array           & 8x8 & 6x6 & 4x4  \\\hline
% beamforming       & Long-term, single stream                                      \\
Connections       & \multicolumn{2}{c|}{Pathloss $< 200$ dB} & No UE-UE \\
\end{tabular}
\end{table}

We compare the sum throughput and sum-utility in the network under our MU-MIMO scheduling model versus two more limited scheduling models in prior works \cite{juanScheduling} and  \cite{gomezITAoptimal}:
\begin{itemize}
 \item \textit{Benchmark 1:} \textbf{1-to-1 mmWave network \cite{juanScheduling}}. Each node can only participate in one link at a time as in Fig. \ref{fig:1t1schedex}. The set of valid schedules is $\mathcal{P}=\{\pp\in\{0,1\}^{L}\textnormal{ s.t. }\sum_{m\in\mathcal{A}(n)}\|p_{n,m}\|_0\leq1, \sum_{n\in\mathcal{A}(m)}\|p_{n,m}\|_0\leq1\}$ . $\mathcal{P}$ is equivalent to the set of all the \textit{matchings} over the graph $\mathcal{G}(\mathcal{N},\mathcal{L})$. mmWave interference is negligible and full power is allocated to a single link for each transmitter. Thus, links are associated with a queue pressure weight $w_{n,m}=\min(R_{n,m},q_{n}^f)\max_f(q_n^f-q_m^f)$ and the MWM algorithm implements the MBP scheduler.
 \item \textit{Benchmark 2:} \textbf{K-to-1 mmWave network \cite{gomezITAoptimal}}. Transmitters can only activate one link but receivers can use MU-MIMO fully as in Fig. \ref{fig:1tkschedex}. The set of valid schedules is $\mathcal{P}=\{\pp\textnormal{ s.t. }\|\pp_{n}\|_0\leq 1\}$ where $\pp_n\triangleq(p_{n,m})_{m\in\mathcal{A}(n)}$ is the subvector of power allocation at $n$. The $\ell_0$ norm $\|\pp_{n}\|_0\leq1$ means that $\pp_n$ must be either all zeros (receiving) or have one element set to $1$ and $|\mathcal{A}(n)|-1$ elements set to zero. In \cite{gomezITAoptimal} we employed a heuristic Message Passing algorithm, while in this paper we simulate the K-to-1 benchmark using MILP with the additional constraint $\|\pp_{n}\|_0\leq1$, which implements the optimal MBP scheduler.
 %that obtained 40\% of the gain versus the 1-to-1 case of the optimal scheduler, computed by brute force exhaustive search.
 \item \textit{Proposed Model:} \textbf{MU-MIMO mmWave network}. Both transmitters and receivers can use MU-MIMO fully as in Fig. \ref{fig:mumimoschedex}. The set of valid schedules is $\mathcal{P}=\{\pp\in\mathbb{R}^L \textnormal{ s.t. } p_{n,m}=s_{n}(1-s_m)/\mathcal{A}(n)\}$. Differently from the benchmarks above, in this case the assumption that power is fixed is a constraint that we \textit{choose}, so that the rates do not vary with $\pp(t)$. This allows to reduce the MBP problem to MWDBSG and makes it tractable using a MILP toolbox \eqref{eq:milproblem}. The link weights for this case, $w_{n,m}=\min(R_{n,m},\xi_{n,m}q_{n}^f)\max_f(q_n^f-q_m^f)$, differ from the weights used for MWM in the 1-to-1 benchmark in the parameter $\xi_{n,m}$.
\end{itemize}
 
\begin{figure}[!t]
 \centering
 \subfigure[1-to-1 - MWM]{\includegraphics[width=.6\columnwidth]{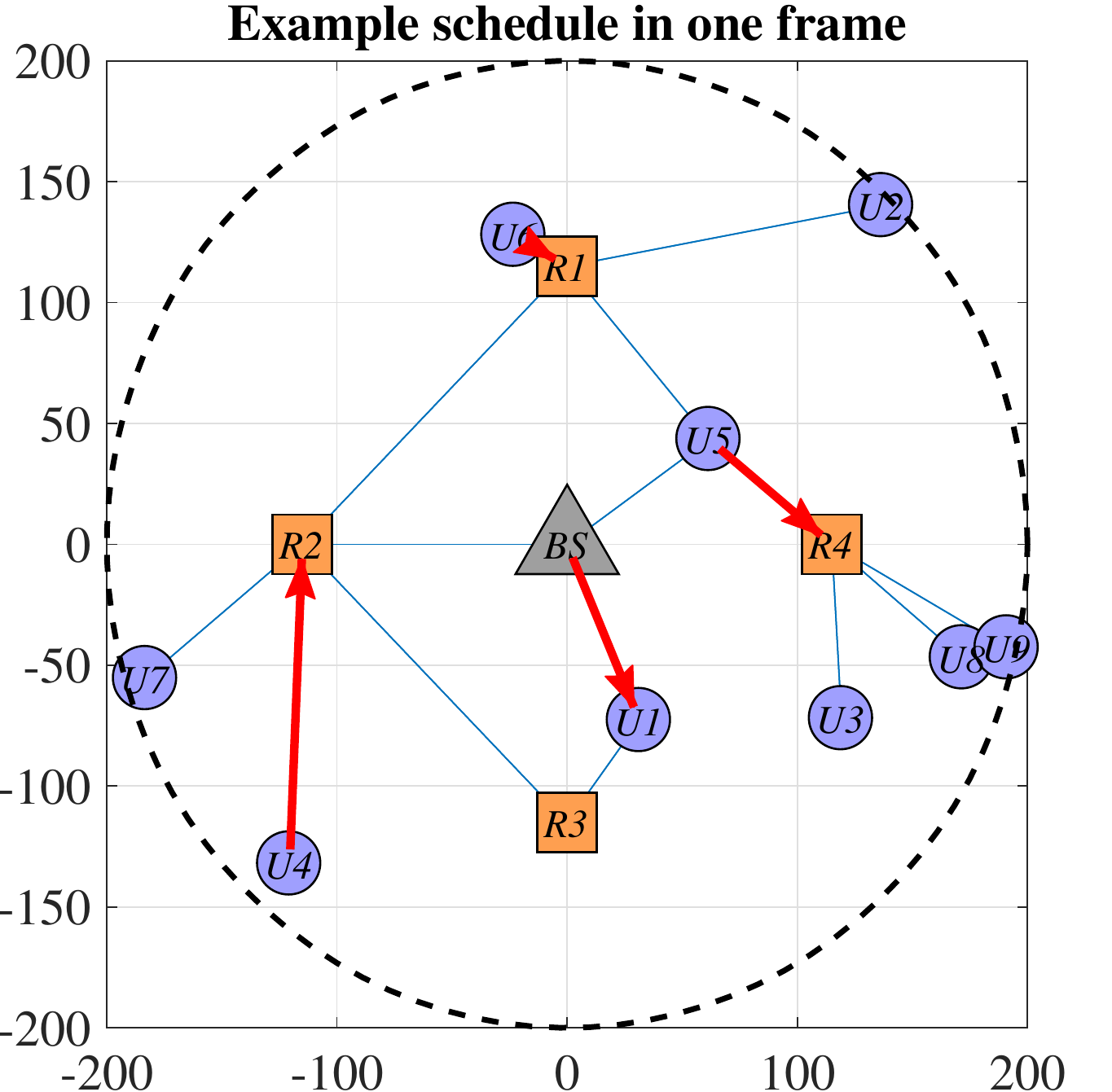}\label{fig:1t1schedex}}
 \hspace{.2in}
 \subfigure[K-to-1 - MILP]{\includegraphics[width=.6\columnwidth]{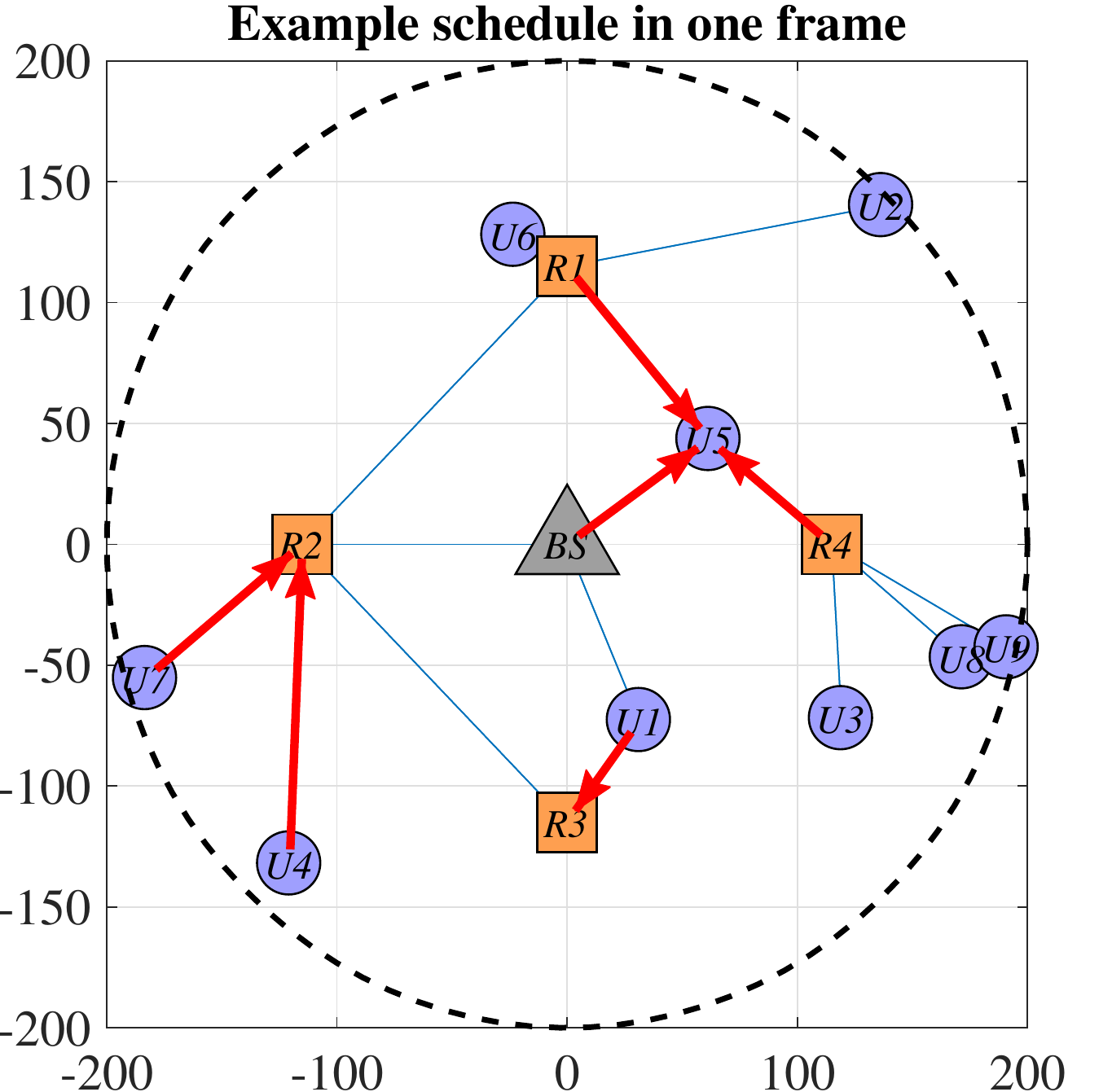}\label{fig:1tkschedex}}
 \hspace{.2in}
 \subfigure[MU-MIMO - MILP]{\includegraphics[width=.6\columnwidth]{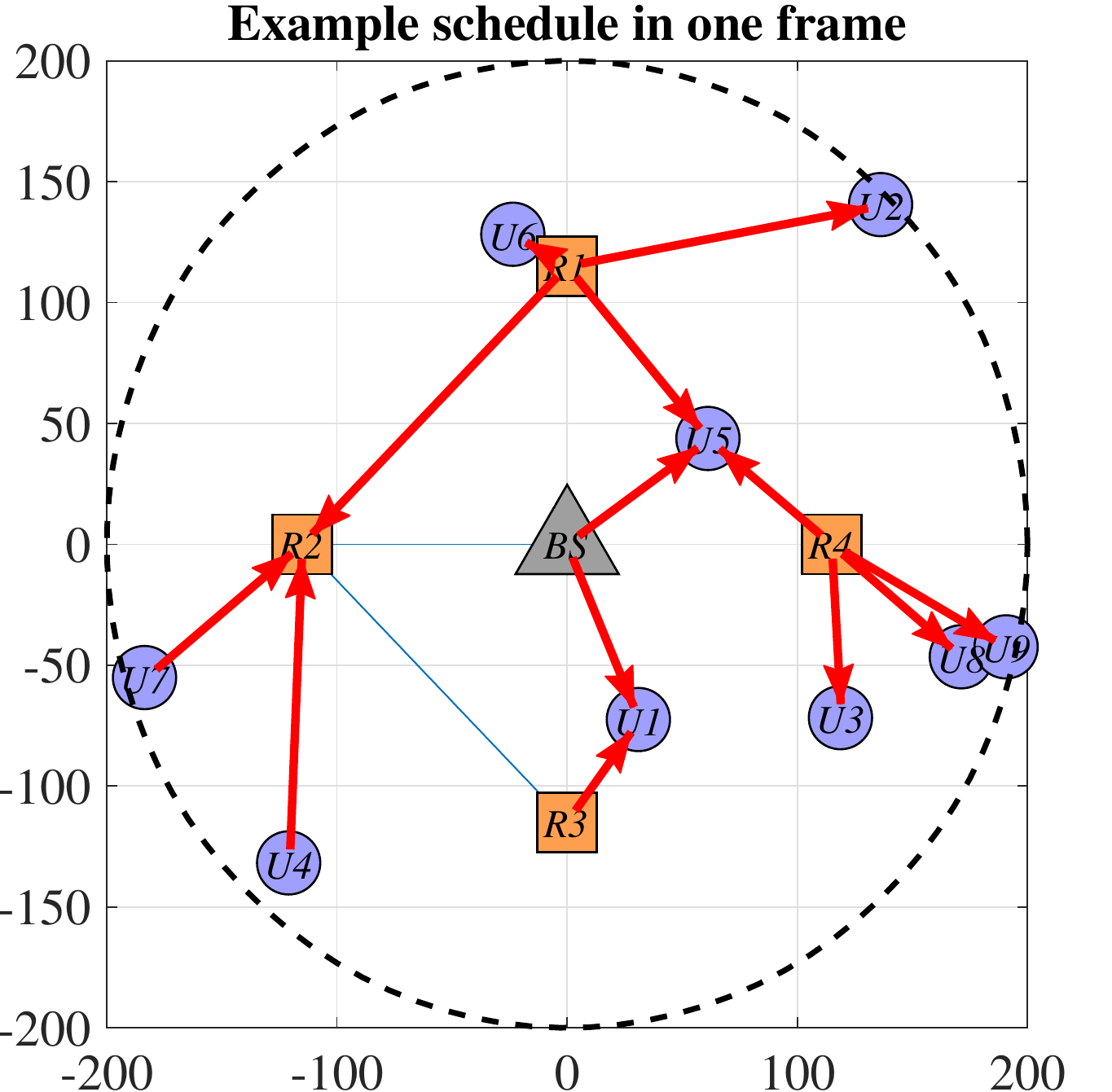}\label{fig:mumimoschedex}}
 \caption{Scheduling constraints in the three multi-hop mmWave picocellular network scheduling models we compare.}
\end{figure}

\subsection{Long-term behavior of the MU-MIMO MBP schedulers}

We first illustrate the behavior of MBP NUM scheduling as a long-term steady-state Markov process network analysis. We consider one of the drops we simulated as an example. We remark that the MBP scheduler is not guaranteed to be a ``good short-term scheduler'' since it experiences a slow start and is not adequate if we want to maximize the performance only in the first frames. Nevertheless, thanks to the aversion of MBP and ANCC to large queue lengths, after the network operates for a sufficient time it converges to a permanent regime where the network state returns, in an average finite time, to states with NUM optimal rates and queues bounded by a constant.

We represent the convergence of the sum of all queue lengths towards a finite upper bound for the three scheduling constraint models in Fig. \ref{fig:allqueues}. We note that the trend in queue lengths does not change significantly between the models. While MBP scheduling guarantees stability, the ANCC algorithm modifies the exogenous arrival rates at the sources in order to find the NUM optimal rates solving \eqref{eq:EPSproblem}. Since the three models change the scheduling constraints, the stability region is increased in the MU-MIMO - MILP model compared to the K-to-1 MILP model, and both of them represent an increase compared to the 1-to-1 MWM model. We see in Fig. \ref{fig:allinputs} that the ANCC converges towards a greater sum rate in the network with MU-MIMO. As the network is stable, traffic is conserved and the mean sum throughput at the destinations in  Fig. \ref{fig:alloutputs} coincides with the ANCC rate at the sources in Fig. \ref{fig:allinputs}, although the instantaneous throughputs at the destinations are more irregular. Finally, the three simulations employ a ``proportional fair'' utility function in ANCC and the division of the sum rates among the users displays similar shapes for all three models in Fig. \ref{fig:fairrates}.

\begin{figure}[!t]
 \centering
 \subfigure[Sum of all queues]{\includegraphics[width=.8\columnwidth]{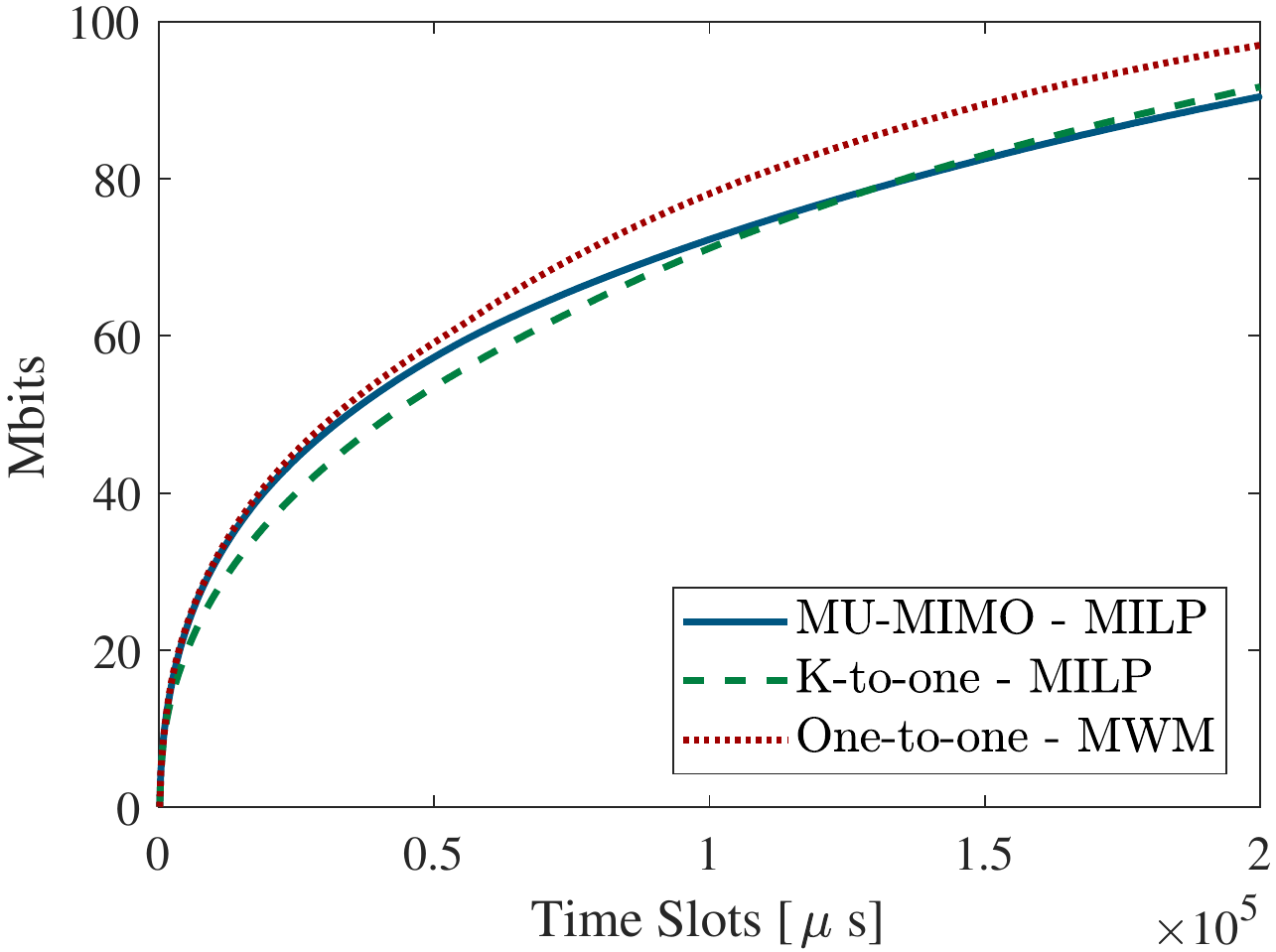}\label{fig:allqueues}}
 \hspace{.2in}
 \subfigure[Sum rate arrival at the sources]{\includegraphics[width=.8\columnwidth]{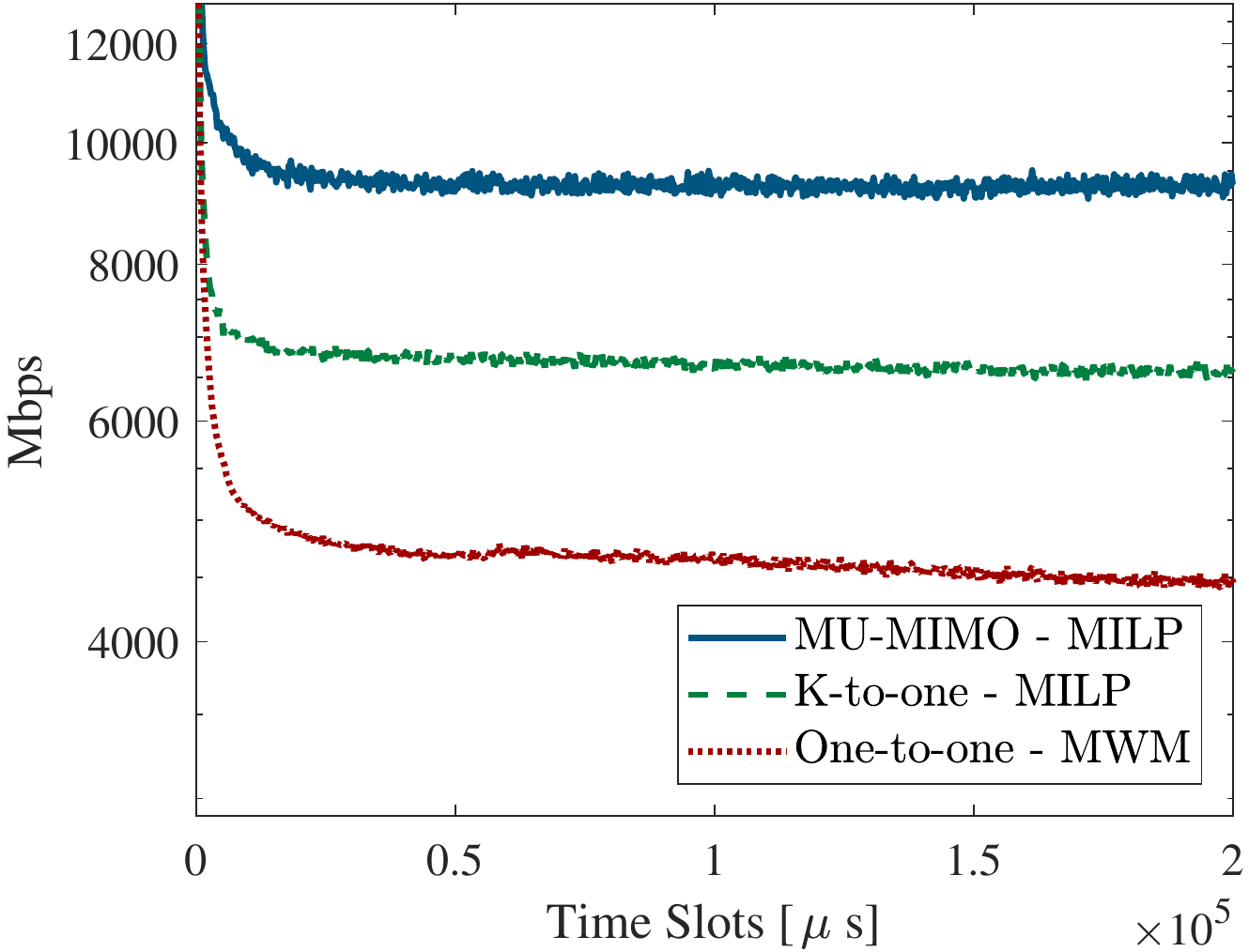}\label{fig:allinputs}}
 \hspace{.2in}
 \subfigure[Sum throughput at the destinations]{\includegraphics[width=.8\columnwidth]{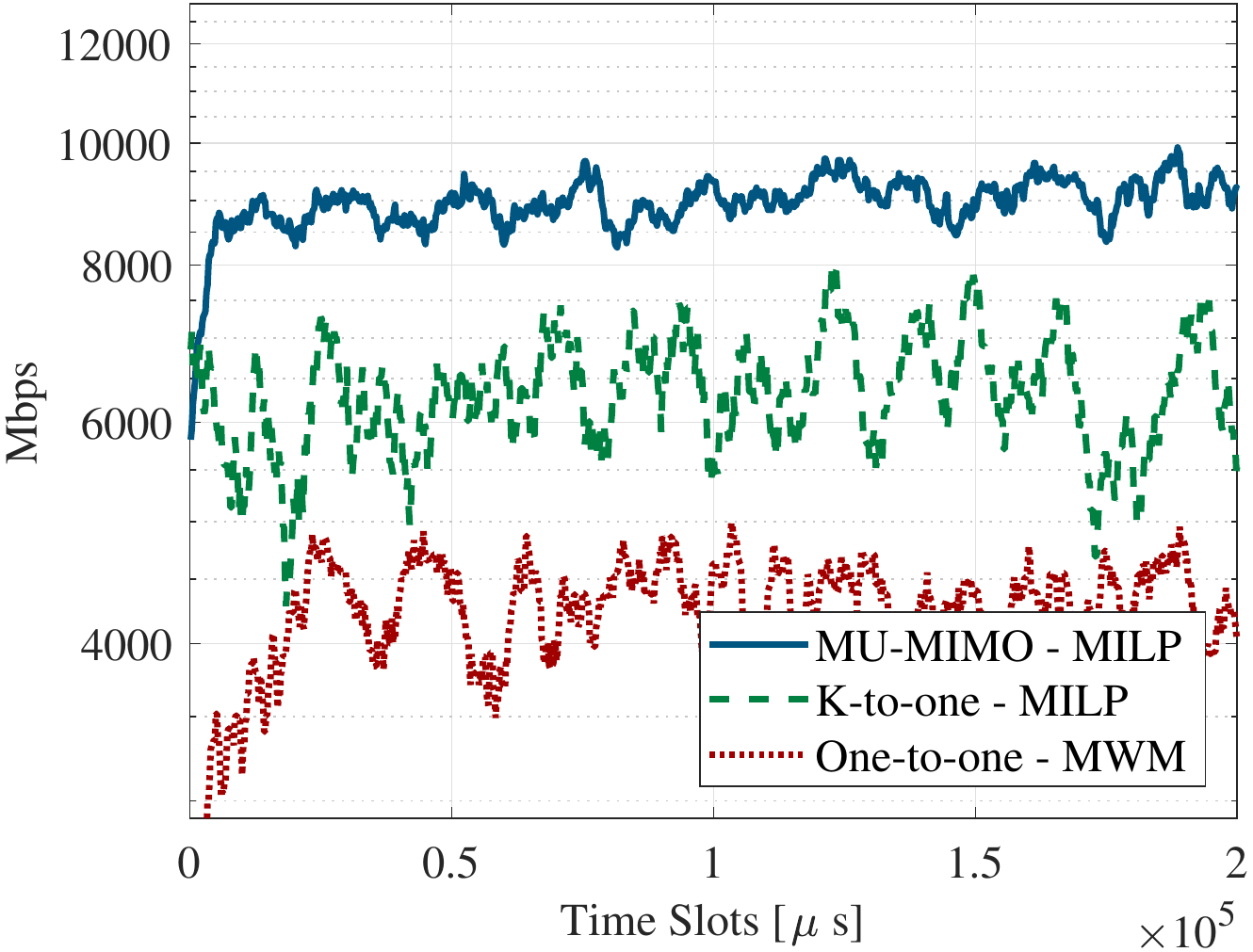}\label{fig:alloutputs}}
 \caption{Variation over time of $\|\q(t)\|_1$, $\|\lambdab(t)\|_1$ and sum throughput arriving at the flow the destinations.}
\end{figure}

\begin{figure}[!t]
 \centering
 \subfigure[1-to-1 - MWM]{\includegraphics[width=.8\columnwidth]{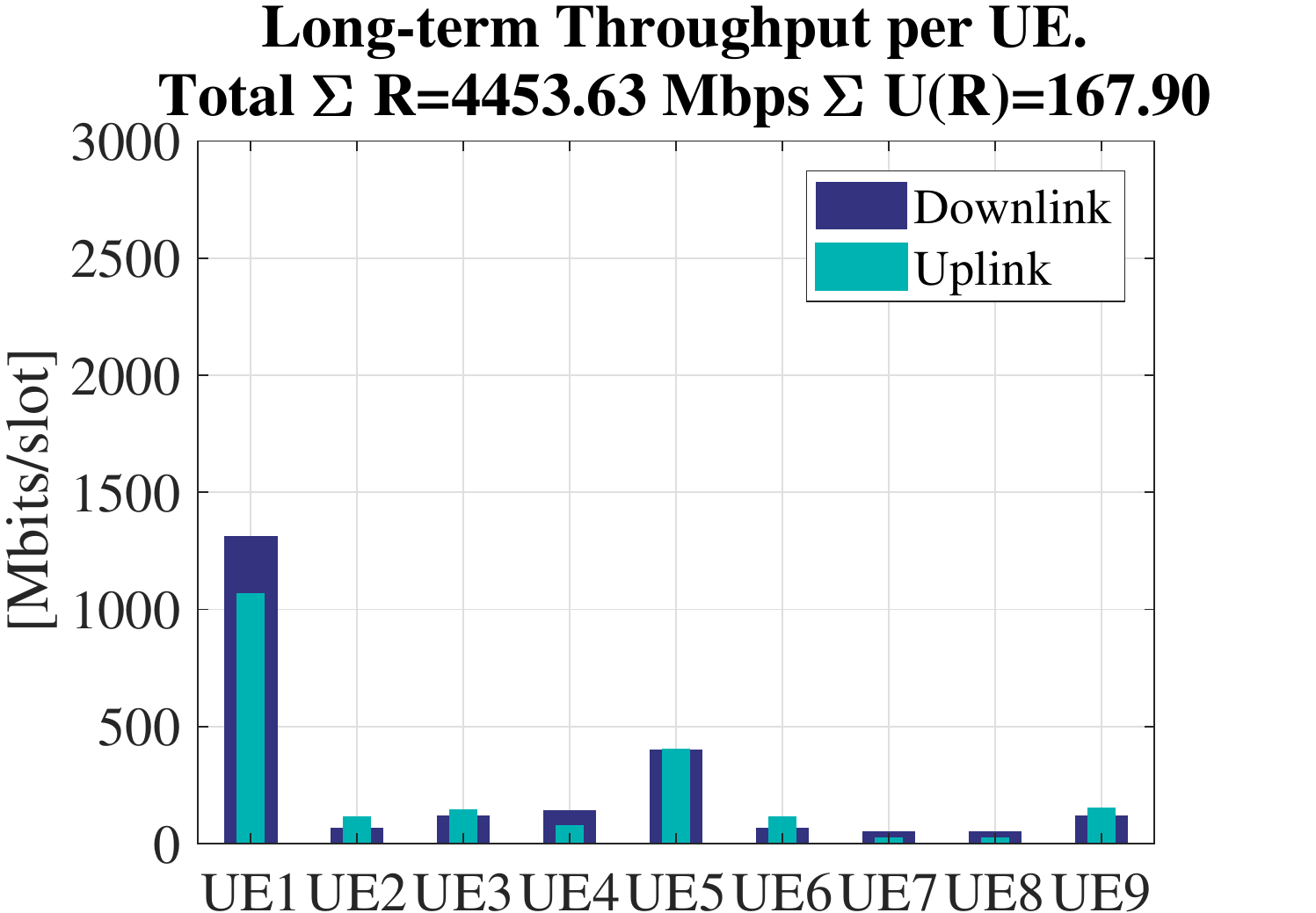}}
 \hspace{.2in}
 \subfigure[K-to-1 - Message Passing]{\includegraphics[width=.8\columnwidth]{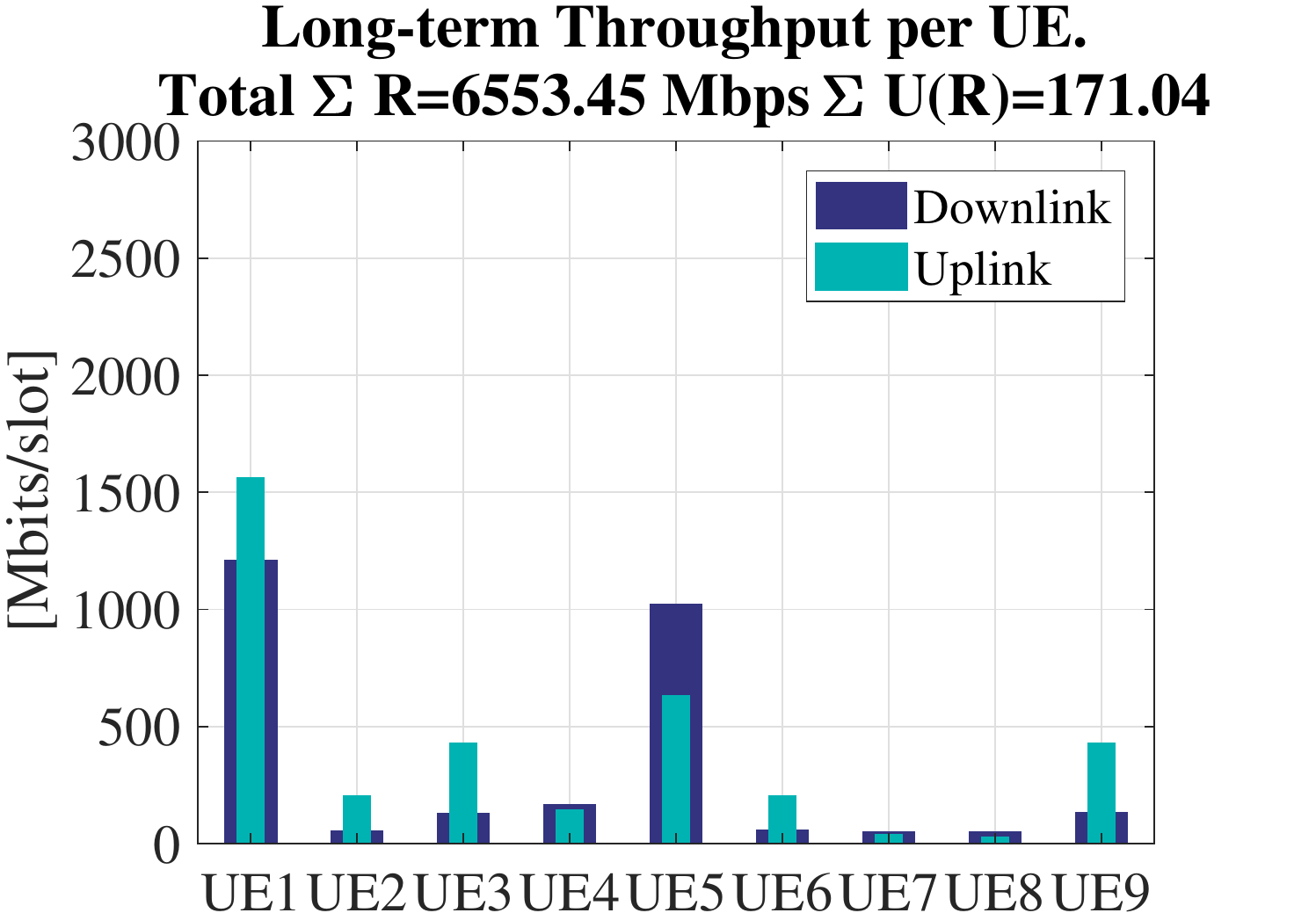}}
 \hspace{.2in}
 \subfigure[MU-MIMO - MILP]{\includegraphics[width=.8\columnwidth]{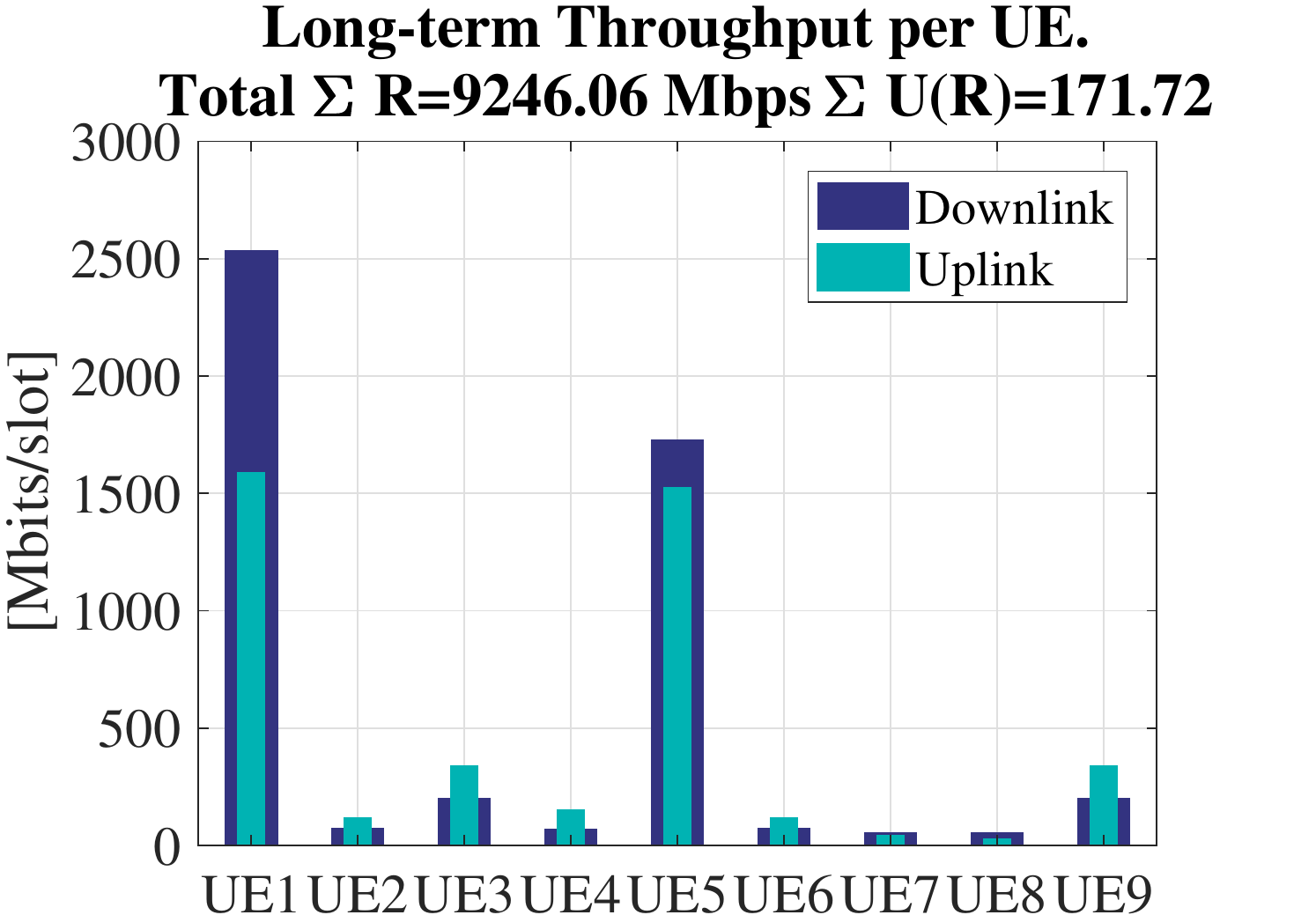}}
 \caption{Rates achieved by each user in the simulation of one realization of the random network. The fairness is similar since the user locations are the same and the rate adaptation gives users a similar division of the network capacity in all models.}
 \label{fig:fairrates}
\end{figure}

\subsection{Characterization of Multi-hop MU-MIMO gain for mmWave Systems}

We repeat our simulations for 50 randomly generated network drops and display the achieved sum rate and sum-network-utility values in Figs. \ref{fig:sumrate} and \ref{fig:sumutility}. Each of our randomly generated drops represents one possible network layout, and the average and standard deviation of the rates across all the drops represents the average performance that mobile users can expect over longer periods of time when the network topology changes. In the left margin of Fig. \ref{fig:sumrate} we have represented the average and standard deviation of the sum rates across all network layouts. We note that the K-to-1 MILP scheduling improved the sum rate in random mmWave multi-hop picocellular network by approximately 90\% over a 1-to-1 MWM model. Moreover, the optimal MU-MIMO - MILP scheduler improved rates even further to a total gain of 160\% over the 1-to-1 model. This improvement stems from the fact that enabling MU-MIMO at the transmitters dramatically increases the spatial multiplexing in the network. In addition to the average rate gain over all network realizations, we note that the gains are consistent and that in almost all simulated networks, individually, MU-MIMO outperformed the K-to-1 and both of these always outperform the 1-to-1 model.

\begin{figure}[!t]
 \centering
 \includegraphics[width=.9\columnwidth]{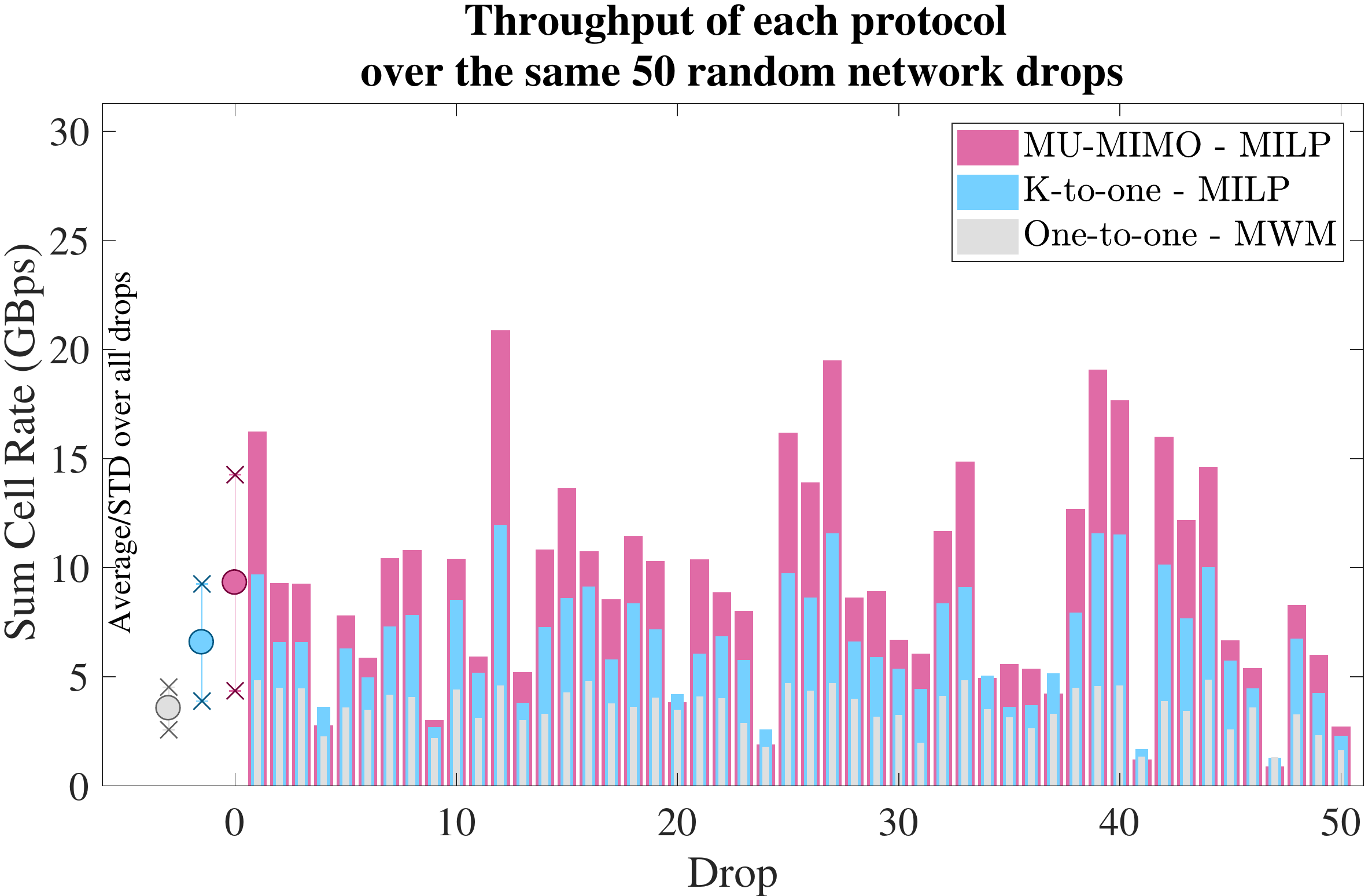}
 \caption{Sum-Rate in 50 randomly generated mmWave picocell layouts.}
 \label{fig:sumrate}
\end{figure}

We remark that the rate gains reported in Fig. \ref{fig:sumrate} do not come at the expense of fairness. In Fig. \ref{fig:sumutility} we show that the total sum network utility improved under MU-MIMO as well, where we have used a proportional fair utility function. Therefore, the new scheduling model does not achieve these additional rates by penalizing poorly located users, but by opening up new spatial multiplexing opportunities in the network scheduling constraints, expanding the network throughput capacity region $\mathcal{X}$, while the long-term average user rates vector $\x\in\mathcal{X}$ is driven to the border of the capacity region in a proportional fair way by the ANCC rate adaptation.

\begin{figure}[!t]
 \centering
 \includegraphics[width=.9\columnwidth]{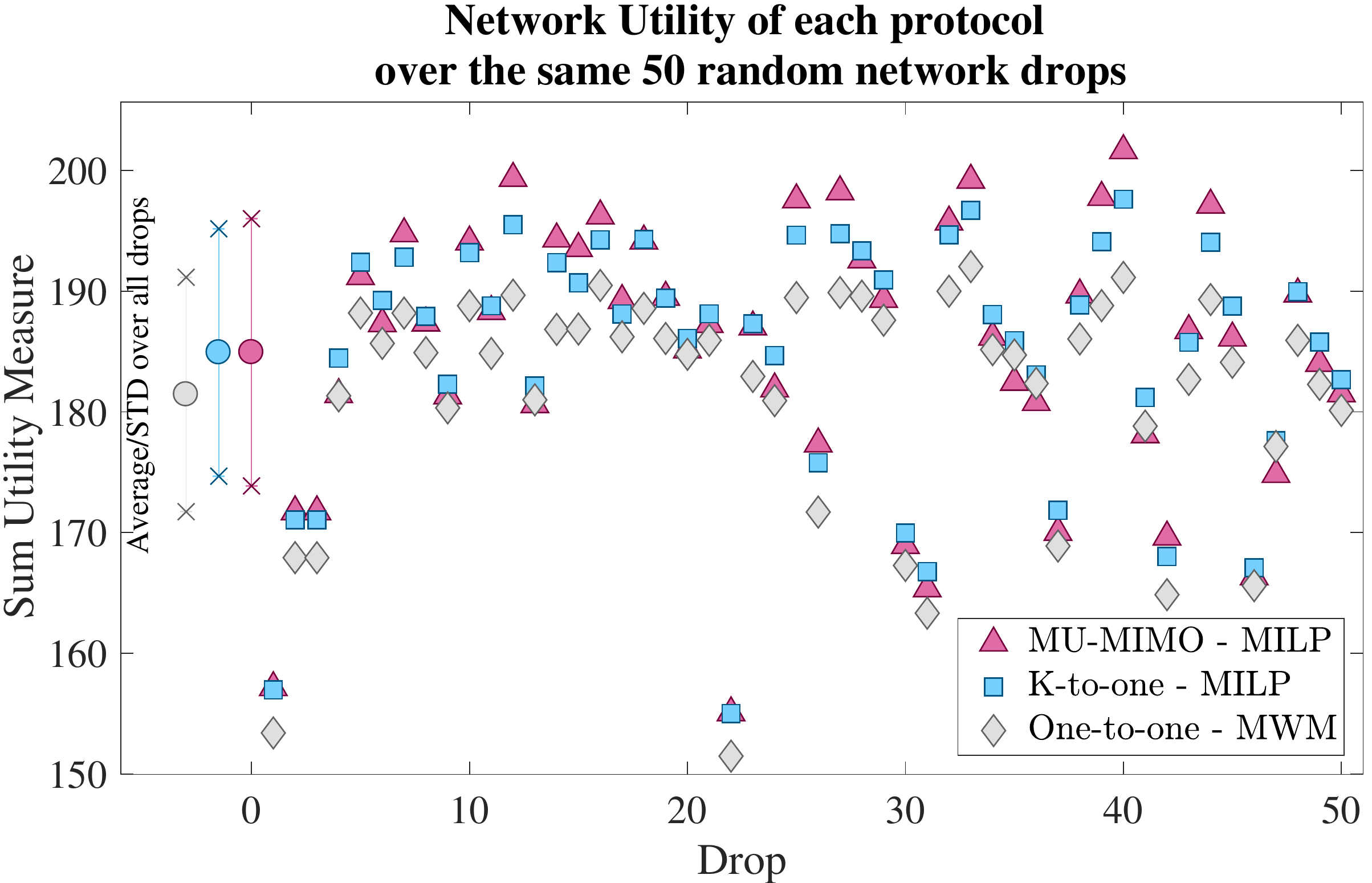}
 \caption{Sum Network Utility achieved with $\mathscr{U}^f(x)=\frac{1}{2}\log(x)$ (proportional fair utility).}
 \label{fig:sumutility}
\end{figure}

In summary, our numerical simulations have shown the importance of providing support for a MU-MIMO physical layer in scheduling for multi-hop mmWave picocellular systems with IAB RNs. While the traditional literature on NUM throughput-optimal multi-hop scheduling has focused on wireless networks with omnidirectional radios and a one-to-one constraint, such as sensor networks, significant increments in network capacity can be achieved if mmWave cellular systems can combine multi-hop and the MU-MIMO physical layer techniques that are frequently used in sub-$6$ GHz single-hop cellular networks, even under a na\"ive power allocation restriction. The extension of our MU-MIMO scheduler to waterfilling and other advanced power allocation techniques can only make the benefits greater.

\subsection{Effect of Interference }
\label{sec:interfeffect}
We note that our theoretical proof that the general form of MBP in \eqref{eq:MBP} is the optimal scheduler remains valid even in the presence of interference. We must assume that  interference is negligible in order to be able to solve \eqref{eq:MBP} using the MILP toolbox. However, one particular \textit{suboptimal} scheduling technique that supports arbitrary interference is Pick and Compare (PaC) without MU-MIMO as proposed in \cite{juanScheduling}. PaC is not optimal in the sense that it does not solve \eqref{eq:MBP}, however for the 1-to-1 case (no MU-MIMO) it has been proven in \cite{juanScheduling} that PaC is throughput optimal (i.e., the queues are stable even though the optimal solution of \eqref{eq:MBP} is not selected).

In Fig. \ref{fig:interfsanitycheck} we simulate the user rates over $10^5$ frames using PaC (no MU-MIMO) assuming interference is negligible (Interference Free) vs the same simulation with the Actual Interference model as in \cite{juanScheduling}. We observe that the user throughput is almost identical in the two scenarios. Therefore, network-level simulation results obtained assuming that there is no interference in the mmWave network are representative, and very similar to, the actual rates that would be observed in a real network where some (small but non-zero) interference is present.

We note that the PaC algorithm of \cite{juanScheduling} does not support MU-MIMO and, in this sense, our ``sanity check'' is not fully comprehensive of the MU-MIMO problem. However, we argue that, since the random channels are independent in our model, the probability that two MU-MIMO links that have a node in common are ``aligned'' (in the sense of having a large cross-beam interference) would not be higher than the probability of the same event for two links without nodes in common in a 1-to-1 simulation. Therefore, our MU-MIMO numerical simulation without interference is, also, an approximated evaluation of the rates in MU-MIMO mmWave IAB picocellular networks with weak interference.

\begin{figure}
  \centering
  \subfigure[Interference Free]{
    \includegraphics[width=.8\columnwidth]{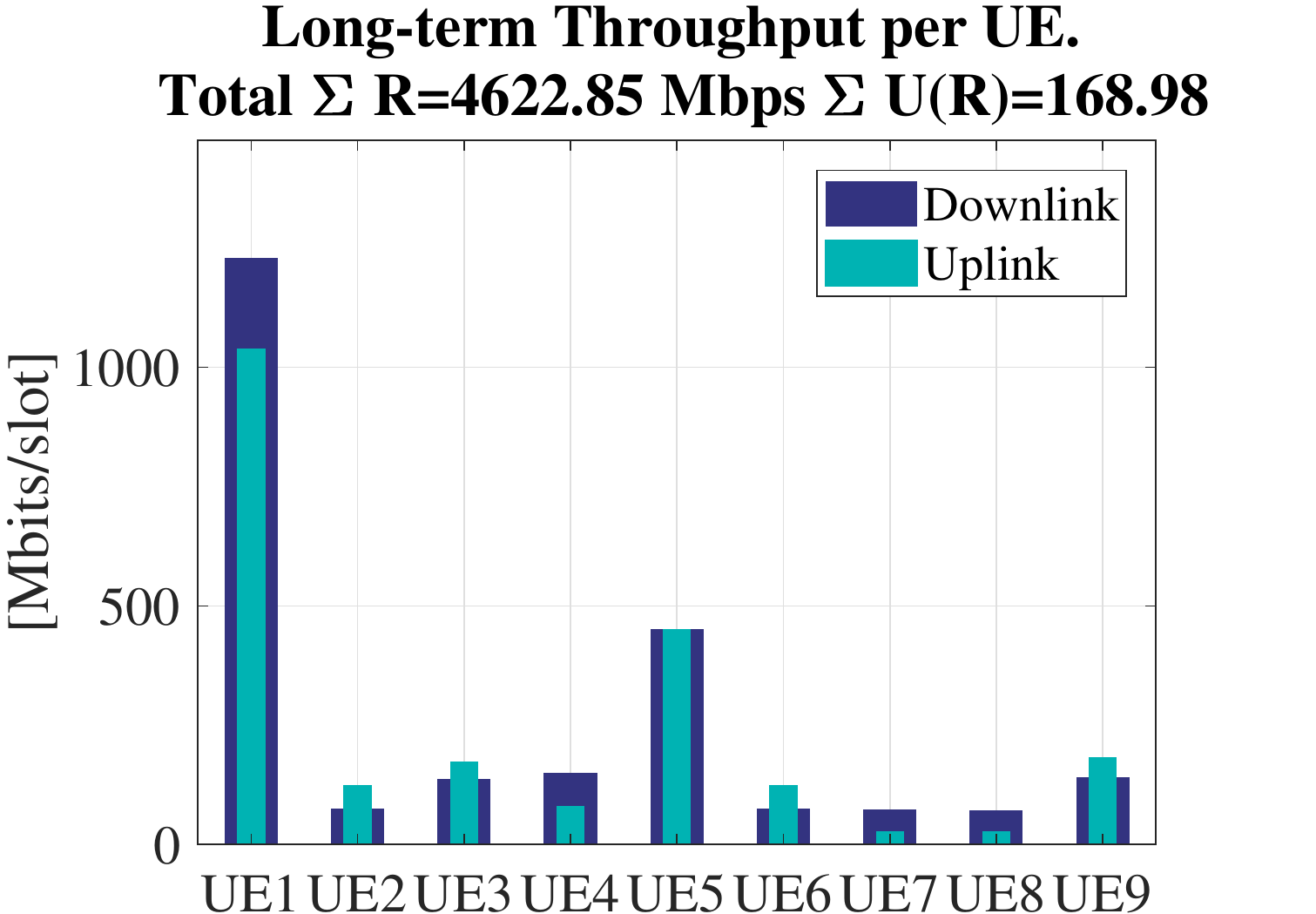}
  \label{fig:nointerf}
  }
  \subfigure[Actual Interference]{
    \includegraphics[width=.8\columnwidth]{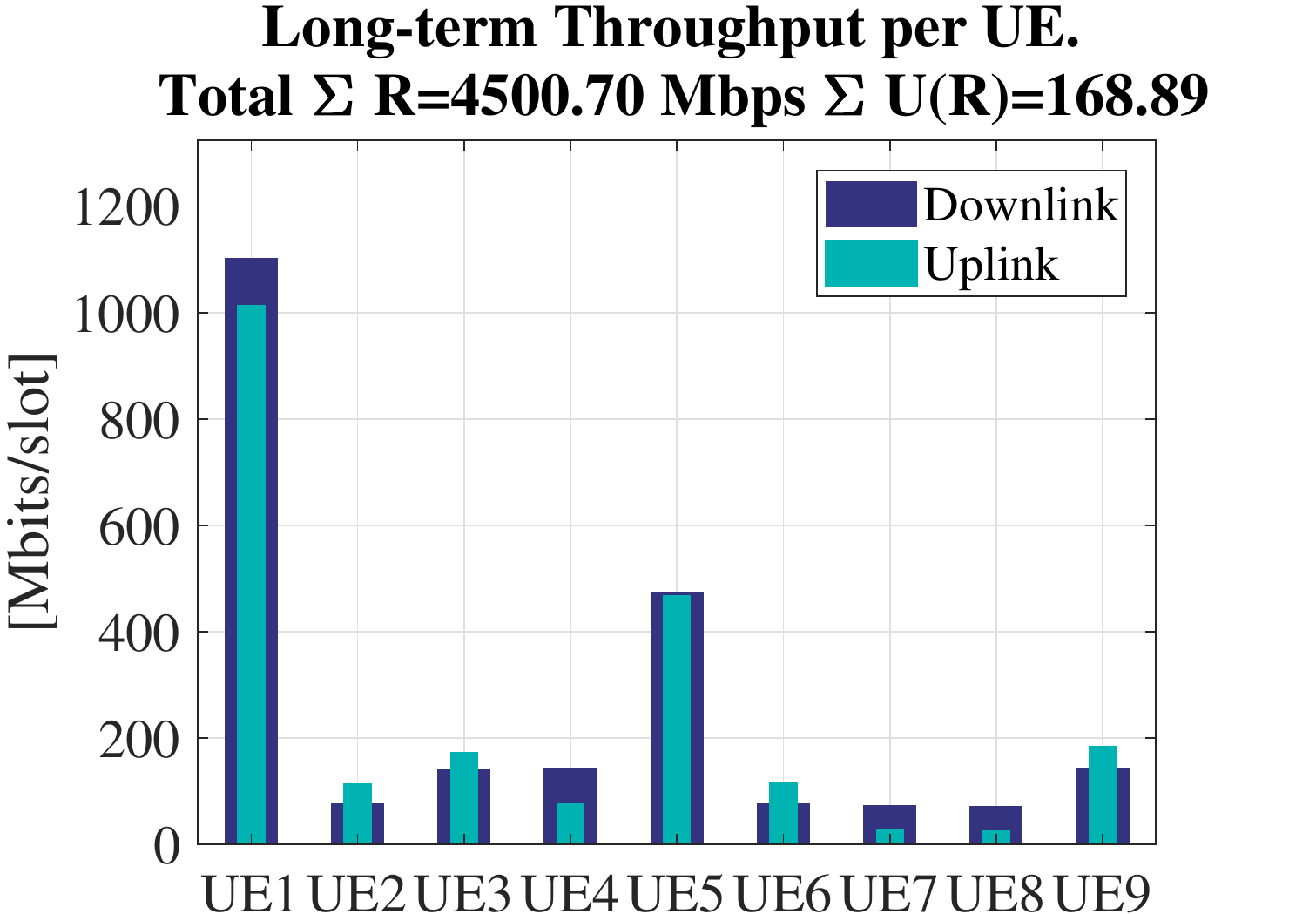}
  \label{fig:actualinterf}
  }
  \caption{Effect of interference in user rates using PaC scheduling (no MU-MIMO) as in \cite{juanScheduling}.}
  \label{fig:interfsanitycheck}
\end{figure}

\section{Conclusions and Extensions}
\label{sec:conclusion}

Future mmWave 5G picocellular networks with IAB RNs require a combined scheduling framework that harmonizes the existing models in MU-MIMO physical layers for conventional cellular systems and in multi-hop one-to-one NUM scheduling literature. We have generalized the classic multi-hop scheduling framework to MU-MIMO and proved the throughput and NUM optimality of a MU-MIMO MBP scheduler with ANCC rate adaptation. The generalized scheduler leads to a problem separation into two subproblems: a state selection for each node, between transmitter and receiver, and a power allocation at the transmitters. Our theoretical proof that MU-MIMO MBP is an optimal scheduler is fully general and admits arbitrary power allocation and non-negligible interference in the link rates. However, under such general conditions the implementation of MU-MIMO MBP as an optimization is challenging.

In the particular case of mmWave, interference is negligible and approximate link rates depend only on the power allocation by the desired-signal transmitter, but not on the power allocations by the interferers. Thanks to this simplification, the power allocation problem can be solved locally at each transmitter. Thus, the MU-MIMO MBP scheduling optimization consists in finding the optimal Directed Bipartite SubGraph of the network. Nonetheless, this problem is still hard  since the power allocations and link rates take different values for each DBSG. 

It can be shown that MU-MIMO MBP is still a throughput optimal scheduler if we incorporate more constraints in the power allocation subproblem, or even if we choose a fixed power allocation strategy. Thanks to this, we make the MU-MIMO MBP scheduler a tractable problem by assuming a fixed-power limitation on the transmitter radios. Under this assumption the MU-MIMO MBP scheduler becomes the Maximum Weighted DBSG problem, which can be converted into a MILP problem and attacked with standard toolboxes. Thus, the extension to MU-MIMO with fixed power allocation of the conventional one-to-one scheduling model leads to the substitution of the Maximum Weighted Matching problem by the MWDBSG problem.

We compared the mmWave multi-hop picocellular network throughput under one-to-one, K-to-one, and MU-MIMO scheduling constraints models, finding that MU-MIMO enables a throughput increase of 160\% over the one-to-one model, whereas our prior attempt in a previous work using a K-to-one scheduling constraint model had at most a 90\% gain. In all simulations proportional-fair rate adaptation was employed, and the sum throughput gains in each cell emerged from the creation of new spatial multiplexing opportunities thanks to the MU-MIMO physical layer, not the penalization of poorly located users. 

In future work we intend to relax a number of assumptions made in this paper. First, the assumption of fixed power allocations is one we had to adopt for the sake of tractability, and the implementation of MU-MIMO MBP scheduling solvers with arbitrary power allocation remains open. Second, we used a MILP toolbox that would require centralized scheduling in practical deployments, and the implementation of nearly-optimal schedulers using distributed message-passing or random pick-and-compare strategies is left for future work. Third, even though references have reported that average rates with and without interference are ``nearly identical,'' the extent of this similarity should be further explored in the future via the implementation of MU-MIMO MBP with arbitrary power allocation in the presence of interference. Fourth, the NUM MBP framework focuses on the steady-state distribution over a large number of frames, and does not offer any performance guarantees during the first few frames of operation. Therefore, the design of ``good short term schedulers'' should also be considered.

\appendices

\section{mmWave Channel and Physical Layer Model}
\label{app:phy}

\begin{figure*}[!t]
 \centering
 \includegraphics[width=.65\textwidth]{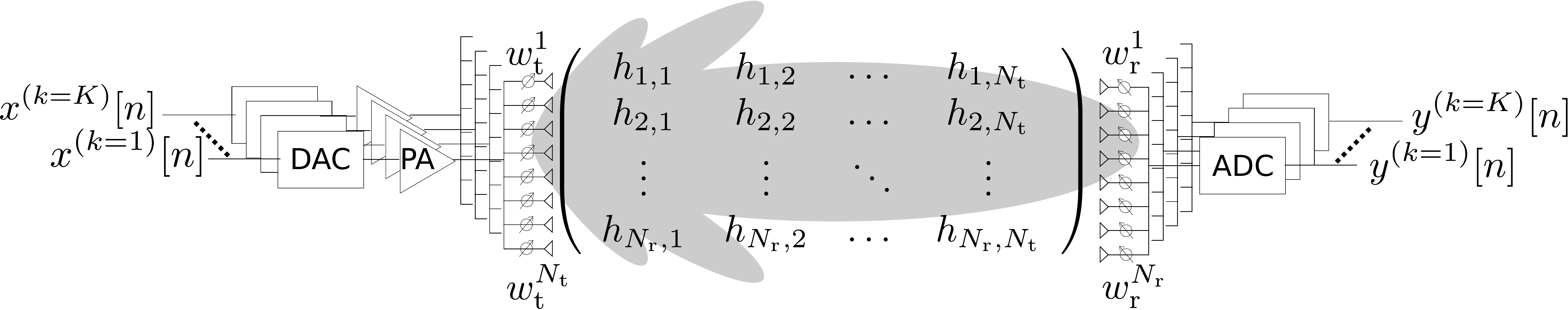}
 \caption{Analog SDM/SDMA scheme with $K$ independent transmitter-beamforming and $K$ independent receiver-beamforming signal ports. Both transmitter and receiver use each port $k\in[1,K]$ for an independent link with a different neighbor node.}
 \label{fig:analogBF}
\end{figure*}

For each node $n$ we assume a mmWave physical layer with $N_a(n)$ antennas. Hybrid analog-digital transceivers are used in transmission and reception, to avoid the high power consumption of fully-digital MIMO \cite{Orhan2015,Mo2016,Abbas2016}. The hybrid analog-digital SDM/SDMA transceiver architectures represented in Fig \ref{fig:analogBF} are assumed. Each node has $K(n)\ll N_a(n)$ independent radiofrequency (RF) digital transmission signal ports connected to a Digital-to-Analog converter, a power amplifier and an analog phased-array block with $N_a$ attenuators and phase-shifters. Each of these $K(n)$ ``RF chains'' is connected to the same antenna array. Reciprocally, at the receiver each node has a common antenna array connected to $K(n)\ll N_a(n)$ parallel RF chains with phase shifters and attenuators, analog combiner units, receive filters and ADCs. We assume that each node has $\mathcal{A}(n)\leq K(n)$ neighbors and both transmitters and receivers can process the signals of all their neighbors simultaneously.

Each transmitter sends up to $K(n)$ independent signals, denoted by $x_{n,1}[t]\dots x_{n,K}[t]$, to $K(n)$ \textbf{different} destinations, denoted by $ m(n,1)\dots m(n,K)$. The use of multiple ports for space division multiplexing towards the same neighbor is of no interest in mmWave because the channel matrix is usually rank-deficient. The transmitter $n$ divides the total transmit power at the node, $P_n$, among the different signals satisfying $\sum_{k=1}^{K}p_{n,m(n,k)}\leq 1$ where $p_{n,m(n,k)}P_n=\Ex{}{|x_{n,k}[t]|^2}$.

We denote by $\mathcal{T}(m)$ the set of transmitters that send signals to receiver $m$, and denote by $\mathcal{R}(n)$ the set of receivers that are targeted by signals from transmitter $n$. The port index number $k$ does not affect link rates, so for notation clarity we drop the index $(n,k)$ in our notation and denote the signal from transmitter $n$ to receiver $m$ as $x_{n,m}[t]$ and its power by $p_{n,m}P_n$. Since we assume that all neighbors have an associated port, the set $\mathcal{T}(m)\setminus \{n\}$ contains all the nodes other than $n$ that introduce interference into the link $(n,m)$, which we write as follows
\begin{equation}
 \begin{split}
  y_{n,m}[t]&=(\w_{n, m}^{\mathrm{r}})^H\Hb_{n, m}\w_{n, m}^{\mathrm{t}}g_{n, m}x_{n,m}[t]\\
%   +z[t]
  &\;+\sum_{\substack{j\in \mathcal{R}(n)\setminus m}}(\w_{n, m}^{\mathrm{r}})^H\Hb_{n, m}\w_{n, j}^{\mathrm{t}}g_{n, m}x_{n,j}[t]\\
  &\;+\sum_{\substack{i\in \mathcal{T}(m)\setminus n}}(\w_{n, m}^{\mathrm{r}})^H\Hb_{i, m}\w_{i, m}^{\mathrm{t}}g_{i, m}x_{i}[t]\\
  &\;+\sum_{i\in \mathcal{T}(m)\setminus n}\sum_{\substack{j\in \mathcal{R}(i)\setminus m n}}(\w_{n, m}^{\mathrm{r}})^H\Hb_{i, m}\w_{i, j}^{\mathrm{t}}g_{i, m}x_{i,j}[t]\\
  &\;+z[t],
  \end{split}
\end{equation}

Here, the first term is the desired transmission from $n$ to $m$, where $n$ transmits the signal $x_{n,m}[t]$ with power $p_{n,m}P_n$,  $g_{n, m}$ is the macroscopic pathloss from $n$ to $m$, $\Hb_{n, m}$ is the $N_a(m)\times N_a(n)$ mmWave \textit{channel matrix} between $n$ and $m$, and $\w_{n, m}^{\mathrm{r}}$ and $\w_{n, m}^{\mathrm{t}}$ are the analog beamforming vectors used by $n$ to transmit to $m$ and by $m$ to receive from $n$, respectively. The second term represents the self-interference caused by $n$ transmitting other signals towards other receivers $\mathcal{R}(n)$, which arrive at $m$ with mismatched transmit beamforming vectors, weakening the interference. The third term represents the signals emitted towards $m$ by other transmitters $\mathcal{T}(m)$, which also leak into the $n,m$ channel as weakened interference due to mismatched receive beamforming vectors. The fourth term contains all other signals sent by neighbors of $m$ ($i\in\mathcal{T}(m)\setminus n$) towards destinations different from $m$ ($j\in\mathcal{R}(i)$). These terms are even weaker due to the use of both transmit and receive beamforming vectors mismatched to the channel matrix $\Hb_{i,m}$. The fifth term is Additive White Gaussian Noise (AWGN) with power spectral density $N_o$. 

Since the beamforming vectors of the three interference terms are not properly matched to the channel matrices $\Hb_{n, m}$ and $\Hb_{i, m}$, whereas the noise power is $N_o$ times a very large bandwidth ($WN_o$), we assume that the power of the interference terms is negligible compared to the noise
\begin{equation}
\begin{split}
\|z[t]\|^2&\gg I_{n,m}(\pp(t))\\
&\triangleq\bigg\|\sum_{\substack{j\in \mathcal{R}(n)\setminus m}}(\w_{n, m}^{\mathrm{r}})^H\Hb_{n, m}\w_{n, j}^{\mathrm{t}}g_{n, m}x_{n,j}[t]\bigg\|^2
\\&
\quad+\bigg\|\sum_{\substack{i\in \mathcal{T}(m)\setminus n}}(\w_{n, m}^{\mathrm{r}})^H\Hb_{i, m}\w_{i, m}^{\mathrm{t}}g_{i, m}x_{i}[t]\bigg\|^2
\\&
\quad+\bigg\|\sum_{\substack{j\in \mathcal{R}(n)\setminus m \\ i\in \mathcal{T}(m)\setminus n}}(\w_{n, m}^{\mathrm{r}})^H\Hb_{i, m}\w_{i, j}^{\mathrm{t}}g_{i, m}x_{i,j}[t]\bigg\|^2\\
%  &\quad\Rightarrow\quad y_{n,m}[t]\simeq (\w_{n, m}^{\mathrm{r}})^H\Hb_{n, m}\w_{n, m}^{\mathrm{t}}g_{n, m}x_{n, m}[t]+z[t]
\end{split}
\end{equation}

Considering this negligible interference simplification we can allow each digital transmission and reception port to design its analog beamforming vectors based on knowledge of the channel matrix $\Hb_{n,m}$ alone in order to maximize the Signal to Noise Ratio (SNR).
\begin{equation}\label{eq:bfdesign}
\begin{split}
 \w_{n, m}^{\mathrm{r}},\w_{n, m}^{\mathrm{t}}=\arg \max& \|(\w^{\mathrm{r}})^H\Hb_{n, m}\w^{\mathrm{t}}\|^2 \\
 \textnormal{s.t. }&\quad\|\w_{n, m}^{\mathrm{r}}\|^2=1,\\
 &\quad\|\w_{n, m}^{\mathrm{t}}\|^2=1
 \end{split}
\end{equation}
We assume that the channel matrices remain constant for the entire duration of the scheduling algorithm. Since the channel is essentially static and beamforming does not depend on interference, each node $n$ can obtain its set of neighbors $\mathcal{A}(n)$ and compute the beamforming vectors $ \w_{m, n}^{\mathrm{r}},\w_{n, m}^{\mathrm{t}}, m \in \mathcal{A}(n)$ before the start of the scheduling process using schemes such as \cite{Barati2015}. 

We model $g_{n, m}$ and $\Hb_{n, m}$ according to \cite[Sec. III]{Akdeniz2014}. The macroscopic pathloss of each link with distance $d(n,m)$ is generated in two steps. First, for each link, a state distribution is generated with three states: Outage (OUT), Line of Sight (LOS) and Non-LOS (NLOS). Second, the pathloss of the link at distance $d$ is calculated depending on its state as follows:
\begin{equation}
\begin{split}\label{eq:pathloss}
 g_{n, m}(\textnormal{dB})&=\begin{cases}
		\infty \\ \quad\textnormal{ w.p. }p_{OUT}=1-\min(1,e^{5.2-0.0334d(n,m)})\\
		61.4+20\log_{10}(d)+ \log \mathcal{N}(0,5.8)\\ \quad\textnormal{ w.p. }p_{LOS}=(1-p_{OUT})e^{-0.0149d(n,m)}\\
		72+29.2\log_{10}(d)+ \log \mathcal{N}(0,8.7)\\  \quad\textnormal{ w.p. }1-p_{OUT}-p_{LOS}
              \end{cases}\\
%               &p_{OUT}\\
%               &p_{LOS}
\end{split}
\end{equation}

The small scale fading matrix $\Hb$ is the sum of a small number of planar waves, where each wave corresponds to one reflection in the scattering environment. This sum is expressed as
\begin{equation}
\Hb_{n,m}=\frac{1}{L}\sum_{k=1}^{N_c}\sum_{\ell=1}^{N_p}g_{k\ell}\ab_r(\theta_r^k+\theta_r^{\ell})\ab_t^{T}(\theta_t^k+\theta_t^{\ell})
\end{equation}
where $N_c\sim \textnormal{Poisson}(1.9)$ is the number of independent scattering clusters.  Each cluster is a bundle of $N_p=20$ paths with similar spatial characteristics but independent gains. The rays of each cluster leave the transmitter with a mean Angle of Departure (AoD) $\theta_r^{k}\sim U[0,2\pi)$ and arrive at the receiver with mean Angle of Arrival (AoA) $\theta_t^{k}\sim U[0,2\pi)$. Each path in the cluster has its individual AoD and AoA separated from the cluster mean angles by a wrapped-Gaussian distribution with mean square angular spread $\theta_{RMS}\sim \textnormal{Exp}(10^o)$. The path angular variations are generated as a wrapped Gaussian $\theta_t^{\ell},\theta_r^{\ell}\sim \textnormal{Wrapped}(\mathcal{N}(0,\theta_{RMS}))$. Finally, for each path in each cluster, the model generates an independent scalar random fading gain $g_{k\ell}\sim\mathcal{CN}(0,1)$ and a spatial signature vector for the antenna arrays that depends on the angles. The $N\times1$ Uniform Linear Array (ULA) with elements separated by half a wavelength has spatial signature vector
$$\ab_\Box(\theta_\Box)=\frac{1}{\sqrt{N_\Box}}\left(0, e^{-j\pi\sin(\theta_\Box)},\dots, e^{-j\pi\sin(\theta_\Box) (N_\Box-1)}\right)^T$$
where the box $\Box$ represents that we can use this expression for both subindices $t$ and $r$. In this paper we assumed the use of $N\times N$ Uniform Planar Arrays (UPA) formed by $N$ ULAs separated by half a wavelength, where the total number of antenna elements is $N_a=N^2$ and the signature is given by $\ab_\Box(\theta_\Box^{Azimuth})\otimes\ab_\Box(\theta_\Box^{Elevation})$ where $\otimes$ is the Kroenecker product, $\theta_\Box^{Elevation}$ and $\theta_\Box^{Azimuth}$ are the azimuth and elevation angles, respectively, and $\Box$ may be either $t$ or $r$.

In this paper we are studying scheduling in a large network without interference, and adopt the beamforming vectors \eqref{eq:bfdesign} as an ``accurate enough'' physical layer model. However, typical implementations are subject to hardware constraints and use limited beamforming codebooks \cite{Xiao2016,Xiao2017a}. Moreover, we have considered independent processing on each signal port of the receiver. In typical implementations an additional layer of $K\times K$ digital MU-MIMO processing is necessary in order to ensure that $I_{n,m}(\pp(t))$ is as weak as necessary, however since mmWave channel matrices are rank-deficient, the ``cross-beam interference'' between mismatched beamforming transmitter-receiver pairs is very small most of the time and the simulation results using \eqref{eq:bfdesign} are quite accurate. In network-level rate and scheduling analyses of mmWave the assumption that interference is small and can be neglected has been extensively supported, e.g., see \cite{Mudumbai2009,juanScheduling}.

\section{Proof of NUM and Throughput Optimality}
\label{app:PAC}

Proofs for results very similar to Theorems \ref{the:MBP} and \ref{the:NUMCC} appear often in NUM literature \cite{Tass,Eryilmaz2007,Kelly1997,Kelly1998,ModianoPower}. The main argument traces back to the analysis of ergodic Markov chains by Tweedie \cite{Tweedie83Markov}. In this appendix we combine ideas from two extensions in the literature. In \cite{Eryilmaz2007,Eryilmaz2010Implementation} the throughput optimality is demonstrated for a multi-hop network with an ``arbitrary interference model scheduling constraint,'' but only for fixed rate links ($r_{n,m}=1$). On the other hand, in \cite{ModianoPower} a proof is given for ``random power allocation,'' which means that link rates are also random, but only for single-hop networks with a specific form of 2-hop scheduling constraints. We combine the techniques used to extend the proof to multi-hop arbitrary constraint and the variable-capacity scenarios to write our proof. In essence, to support multiple links at once/arbitrary constraints we take into account the maximum graph degree $A_{\max}$, whereas to support random power allocation/variable rate links we assume that the power allocations have a finite maximum value, so we can upper bound all the link rates by the supremum link rate of the network.

The proof begins by considering the joint variable $\uu(t)=(\q(t),\R(t))$ to represent a state of the network and scheduling system. This joint variable follows a Markov chain with some state space $\uu(t)\in\mathcal{M}$. We consider a Lyapunov function of the state of the system defined as 
\begin{equation}
\mathscr{L}(\uu(t))=
% \underset{\mathscr{L}(t)}{\underbrace{\sum_{n,f}|q_{n}^{f}(t)|^2}}+\underset{\mathscr{L}_2(t)}{\underbrace{\sum_{n,m,f}\left[([r_{n,m}^f]_{\textnormal{MBP}}-r_{n,m}^f)(q_{n}^{f}(t)-q_{m}^{f}(t))\right]^2}}
% \end{equation}
% 
% The first term bounds the total queue length in the system
% \begin{equation}
% \mathscr{L}(t)=
\|\q(t)\|^2=\q^T(t)\q(t)\geq \sum_{n,f} q_{n}^{(f)}(t)
\end{equation}

The goal is to show that if there is a finite highest link capacity, $R_{\max}\triangleq \max_{n,m} r_{n,m}(t)|_{p_{n,m}(t)=1}<\infty$, and finite highest graph degree, $A_{\max}=\max_n|\mathcal{A}(n)|<\infty$, the average Lyapunov drift
$$\mathscr{D}(\uu(t))=\Ex{t}{\mathscr{L}(\uu(t+1))-\mathscr{L}(\uu(t))|\uu(t)}$$
is always negative if $\uu(t)$ is contained in a certain set $\mathcal{B}^c$, where $\mathcal{B}\cup \mathcal{B}^c\triangleq\mathcal{M}$. When this holds, the Foster-Lyapunov criterion \cite{Tweedie83Markov} establishes that the irreducible Markov chain is positive recurrent, i.e., the average return time to states  in the set $\mathcal{B}$ is finite.

We prove this criterion for $\mathcal{B}$ such that the queues are bounded $\mathcal{B}\subset\{\uu(t):\|\q\|_1\leq \sqrt{B}\}$. Moreover, we define $\mathcal{B}$ such that the difference between the achieved long term average utility and \eqref{eq:EPSproblem} vanishes. Thus, the network can be modeled by a Markov process that on average takes a finite number of transitions to return to states with bounded queues (network stability) and maximum network utility.
% the states that select a schedule that is very far off of the optimal have a stochastic tendency to select better schedules next; both behaviors occurring at the same time stabilize the network and guarantee NUM.

% \subsection{Analysis of $\mathscr{L}(t)$}

% We compute the 1-step Lyapunov drift of the function $\mathscr{L}(\uu(t))$ as
% $$\mathscr{D}(t)=\Ex{}{\mathscr{L}(\uu(t+1))-\mathscr{L}(\uu(t))}$$

We introduce \eqref{eq:qupdate} into $\mathscr{D}(\uu(t))$, and drop the time index $(t)$ for clarity:
$$\mathscr{D}=\Ex{}{\|\q+\D(\R^T-\R)\one_{NF,1}+\D\ab\|^2-\|\q\|^2}$$

Next, we expand the square sum and set $\D=\I$ to upper bound $\mathscr{D}$, producing
\begin{equation}
\begin{split}
   \mathscr{D}&\leq\underset{ 0}{\underbrace{\Ex{}{\|\q\|^2-\|\q\|^2}}}+\underset{\leq N^2R_{\max}+NA_{\max}R_{\max}}{\underbrace{\Ex{}{\|(\R^T-\R)\one_{NF,1}+\ab\|^2}}}
   \\   & \quad
   +2\Ex{}{\q^T\left[(\R^T-\R)\one_{NF,1}+\ab\right]}\\
  \end{split}
\end{equation}

We add and subtract $\Ex{}{2V\one_{1,NF}\mathscr{U}(\ab)}$. This can be omitted for throughput optimality with fixed rate $\lambdab\in\mathcal{X}$, i.e., to prove Theorem \ref{the:MBP}. We define $C_1\triangleq N^2R_{\max}+NA_{\max}R_{\max}$ so
\begin{equation}
\begin{split}
   \mathscr{D}&<C_1+\Ex{}{2V\one_{1,NF}\mathscr{U}(\ab)}
   \\&\quad
   -\Ex{}{\Ex{}{2V\one_{1,NF}\mathscr{U}(\ab)}-2\q^T\ab}
   \\&\quad
   +2\Ex{}{\q^T(\R^T-\R)\one_{NF,1}}\\
  \end{split}
\end{equation}

We introduce $\x^{V}$, the solution to the approximate problem \eqref{eq:EPSproblem}. By definition, $\x^{V}$ maximizes the absolute value of the third term. Moreover, the MBP scheduler selects the rate matrix $\R_{\textnormal{MBP}}$.
%The drift is averaged conditioned on $\uu(t)$; on a given state $\uu(t)$ the value of $\mathscr{L}_2(t)$ is deterministic, so we clear the $\Ex{}{}$ from that term, obtaining
\begin{equation}
\small
\begin{split}
   \mathscr{D}&\leq C_1+\Ex{}{2V\one_{1,NF}\mathscr{U}(\ab)}
   \\   &\quad
   -\Ex{}{\Ex{}{2V\one_{1,NF}\mathscr{U}(\x^{V})}-2\q^T\x^{V}}
   \\   &\quad
   +2\Ex{}{\q^T(\R_{\textnormal{MBP}}^T-\R_{\textnormal{MBP}})\one_{NF,1}}\\
%    &+2\sqrt{\mathscr{L}_2(t)}\\
%   \end{split}
% \end{equation}
% we then arrange this as 
% \begin{equation}
% \begin{split}
%    \mathscr{D}&\leq
   &=
   C_1+\Ex{}{2V\one_{1,NF}\mathscr{U}(\ab)-2V\one_{1,NF}\mathscr{U}(\x^{V})}
   \\   &\quad
   +2\Ex{}{\q^T(\R_{\textnormal{MBP}}-\R_{\textnormal{MBP}}^T)\one_{NF,1}-\q^T\x^{V}}
%    \\
%    &+2\sqrt{\mathscr{L}_2(t)}\\
  \end{split}
\end{equation}

Assuming that $\x^{V}\in\text{Int}\{\mathcal{X}\}$, there exists some convex linear combination of link rates $\R^V\in \text{Co}\{\R(\pp(t))\;\forall\pp(t)\in\mathcal{P}\}$ such that the net traffic in the source is $((\R^V)^T-\R^V)\one_{NF,1}-\epsilon \one_{NF,1}=\x^{V}$ for some small positive $\epsilon$. By contradiction, if no such linear combination existed, then at least one queue would grow to infinity, the network would not be stable, and $\x^{V}\notin\mathcal{X}$.

Finally, we note that $\min_{\pp} \q^T(\R^T-\R)\one_{NF,1}$ is equivalent to $\max_{\pp} \q^T(\R-\R^T)\one_{NF,1}$ and both are MBP in \eqref{eq:MBP}, therefore $\q^T(\R_\text{MBP}^T-\R_\text{MBP})\one_{NF,1}<\q^T((\R^V)^T-\R^V)\one_{NF,1}$, and
\begin{equation}
\begin{split}
\label{eq:finalD1}
   \mathscr{D}&\leq C_1+\Ex{}{2V\one_{1,NF}\mathscr{U}(\ab)-2V\one_{1,NF}\mathscr{U}(\x^{V})}
   \\   &\quad
   -\Ex{}{\epsilon \q^T\one_{NF,1}}\\
  \end{split}
\end{equation}

% This captures the essential simplifications on specific terms in Theorem \ref{the:MBP}. 
% \begin{itemize}
%  \item 
 The second term is introduced for Theorem \ref{the:NUMCC} only. If we ignore it, we obtain the proof of Theorem \ref{the:MBP} for static traffic contained in the throughput capacity region: We use that $(\q^T\one_{NF,1})^2\geq\|\q\|^2=\mathscr{L}(\uu(t))$. Define the constant $B=\left(\frac{C_1}{\epsilon}\right)^2$, if $\mathscr{L}(\uu(t))>B$ then $\q^T\one_{NF,1}>\sqrt{B}$ and $\mathscr{D}$ is negative. Thus, $\uu(t)\to\uu(t+1)$ is positive recurrent with $\mathcal{B}\subset\{\uu(t):\mathscr{L}(\uu(t))\leq B\}$.
 %  \item The third term is related to the use of random or sub-optimal approximations of \eqref{eq:MBP}. If the MBP scheduler is directly employed, the term $\mathscr{L}_2(t)$ is zero and the proof is simplified. 
% \end{itemize}

For Theorem \ref{the:NUMCC} we need to prove NUM as well as stability. The second term is upper bounded by $2VN\mathscr{U}(R_{\max}A_{\max})$, and the proof of stability is the same but considering the expression
\begin{equation}
 \label{eq:D1compact}
 \mathscr{D}\leq C_1+2VN\mathscr{U}(R_{\max}A_{\max})-\epsilon \|\q\|_1
 %+2\sqrt{\mathscr{L}_2(t)}
\end{equation}

The proof is thus almost complete. To show that the rate is optimal we revise \eqref{eq:finalD1} and note that $\Ex{}{\epsilon \q^T\one_{NF,1}}>0$. So the drift also becomes negative if $ \Ex{}{\one_{1,NF}\mathscr{U}(\x^{V})-\one_{1,NF}\mathscr{U}(\ab)}>\frac{C_1}{V}$. Thus the Foster-Lyapunov criterion is met and $\uu(t)\to\uu(t+1)$ is positive recurrent with
$$\mathcal{B}\subset\left\{\uu(t):
\begin{array}{c}
  \|\q(r)\|^2\leq \left(\frac{C_1+2VN\mathscr{U}(R_{\max}A_{\max})}{\epsilon}\right)^2\\
  \Ex{}{\one_{1,NF}\mathscr{U}(\x^{V})-\one_{1,NF}\mathscr{U}(\ab)}\leq\frac{C_1}{V}
\end{array}\right\}.$$

Here, the arbitrarily small number $\epsilon>0$ represents that the queue length bound is always finite, but becomes higher when the rate vector gets closer to the border of the throughput capacity region, and the tuning number $V$ lets us choose how close to the optimal rates we want to get while we have to take into account that the queues are allowed to grow even more as we increase $V$.

% \bibliographystyle{IEEEtran}
% \bibliography{directiveMAC,WidebandScaling,LTE,publishedWorks,mmWave,massiveMIMO,Math,Schedulers}


% \vspace{-4in}
\begin{IEEEbiography}
[{\includegraphics[width=1in,height=1.25in,clip,keepaspectratio]{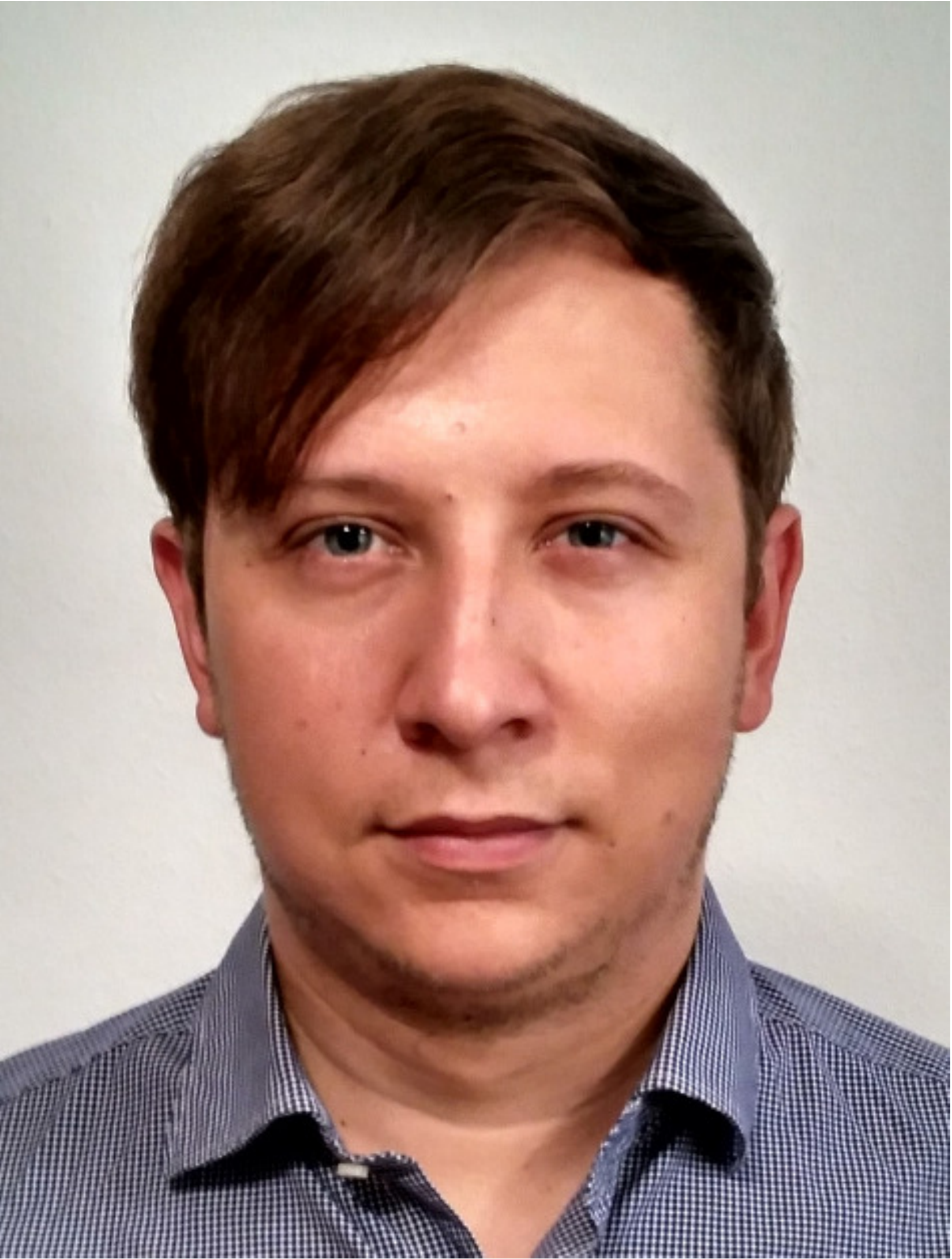}}]
{Felipe G\'omez-Cuba (M'11)}
received the Ingeniero de Telecomunicaci\'on degree in 2010, M.Sc in Signal Processing Applications for Communications in 2012, and a PhD degree in 2015 from the University of Vigo, Spain. He worked as a researcher in the Information Technologies Group (GTI), University of Vigo, (2010--2011), the Galician Research and Development center In Advanced Telecommunications (GRADIANT),  (2011--2013), the NYUWireless center at NYU Tandon School of Engineering (2013--2014) and in University of Vigo with the FPU grant from the Spanish MINECO (2013--2016). He received a Marie Curie Individual Fellowship - Global Fellowship with the Dipartimento d'Engegneria dell'Informazione, University of Padova, Italy (2016-2019) and he was a postdoctoral scholar with the Department of Electrical Engineering, Stanford University, USA (2016-2018). He has been awarded a Distinguished Researcher position with the Beatriz Galindo fellowship in the Departamento de Tecnología Electrónica y de las Comunicaciones (TSC), University of Vigo (2020-2024). His main research interests are new paradigms in wireless networks such as cooperative communications, ultra high density cellular systems, wideband communications and massive MIMO.
\end{IEEEbiography}

% \vspace{-4in}
\begin{IEEEbiography}
[{\includegraphics[width=1in,height=1.25in,clip,keepaspectratio]{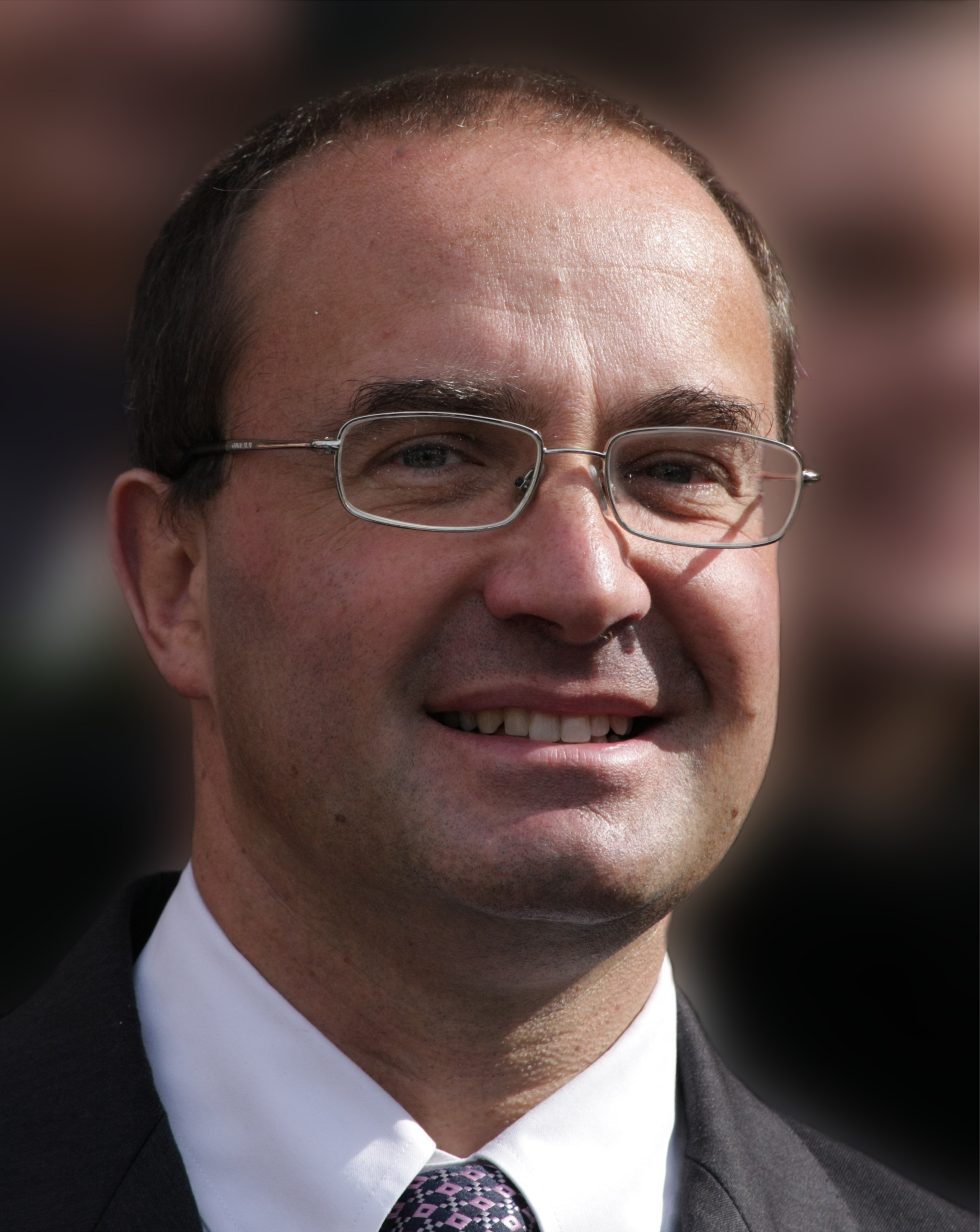}}]
{Michele Zorzi (F'07)}
received the Laurea and Ph.D. degrees in electrical engineering from the University of Padova in 1990 and 1994, respectively. From 1992 to 1993, he was on leave at the University of California at San Diego (UCSD). In 1993, he joined the Faculty of the Dipartimento di Elettronica e Informazione, Politecnico di Milano, Italy. After spending three years with the Center for Wireless Communications, UCSD, he joined the School of Engineering, University of Ferrara, Italy, in 1998, where he became a Professor in 2000. Since 2003, he has been with the Faculty of the Information Engineering Department, University of Padova. His current research interests include performance evaluation in mobile communications systems, WSN and Internet of Things, cognitive communications and networking, vehicular networks, 5G mm-wave cellular systems, and underwater communications and networks. He was a recipient of several awards from the IEEE Communications Society, including the Best Tutorial Paper Award in 2008, the Education Award in 2016, and the Stephen O. Rice Best Paper Award in 2018. He was the Editor-In-Chief of the IEEE W IRELESS COMMUNICATIONS from 2003 to 2005, the IEEE
TRANSACTIONS ON COMMUNICATIONS from 2008 to 2011, and the IEEE TRANSACTIONS ON COGNITIVE COMMUNICATIONS AND NETWORKING from 2014 to 2018. He has served as a Member-at-Large for the Board of Governors of the IEEE Communications Society from 2009 to 2011 and the
Director of Education from 2014 to 2015.

\end{IEEEbiography}

\end{document}